\documentclass[12pt]{article}
\title{Clapeyron-type  theorems in nonlinear elasticity}
\author{Yury Grabovsky\thanks{Department of Mathematics, Temple University, Philadelphia, PA 19122, USA} \and Lev Truskinovsky\thanks{PMMH, CNRS -- UMR 7636, ESPCI, PSL,  75005 Paris, France}}

\usepackage{amssymb,bm} 
\usepackage{amsmath} 
\usepackage{amsthm} 
\usepackage{graphicx} 
\usepackage{tgtermes}
\usepackage{BOONDOX-cal} 
\usepackage[active]{srcltx} 
\usepackage{pdfsync} 
\usepackage{hyperref} 
\hypersetup{pdfborder=0 0 0}
\usepackage[top=1.05in,left=1.05in,right=1.05in,bottom=1.05in]{geometry} 

\counterwithin{equation}{section} 
\newtheorem{theorem}{{\sc Theorem}}[section]

\newtheorem{lemma}[theorem]{{\sc Lemma}}

\newtheorem{remark}[theorem]{Remark}

\newtheorem{definition}[theorem]{Definition}

\newcommand{\bb}[1]{\mathbb{ #1}}

\bmdefine\Bone{1}

\newcommand{\Sym}{\mathrm{Sym}}

\newcommand{\dOm}{\partial\Omega}
\newcommand{\eqv}{\Longleftrightarrow}
\newcommand{\bra}[1]{\overline{#1}}

\newcommand{\Trc}{\mathrm{Tr}\,}

\newcommand{\cof}{\mathrm{cof}}

\newcommand{\tns}[1]{#1\otimes #1}
\newcommand{\hf}{\displaystyle\frac{1}{2}}
\newcommand{\nth}[1]{\displaystyle\frac{1}{#1}}

\newcommand{\dif}[2]{\displaystyle\frac{\partial #1}{\partial #2}}
\newcommand{\Grad}{\nabla}
\newcommand{\Div}{\nabla \cdot}

\newcommand{\Md}{\partial}
\renewcommand{\Hat}[1]{\widehat{#1}}
\newcommand{\Tld}[1]{\widetilde{#1}}
\newcommand{\hess}[2]{\displaystyle\frac{\partial^2 #1}{\partial #2^2}}
\newcommand{\mix}[3]{\displaystyle\frac{\partial^2 #1}{\partial #2\partial #3}}
\newcommand{\av}[1]{\left\langle #1 \right\rangle}

\def\XXint#1#2#3{{\setbox0=\hbox{$#1{#2#3}{\int}$ }
\vcenter{\hbox{$#2#3$ }}\kern-.6\wd0}}

\newcommand\myatop[2]{\genfrac{}{}{0pt}{}{#1}{#2}}

\newcommand{\jump}[1]{\lbrack\!\lbrack #1 \rbrack\!\rbrack}
\newcommand{\lump}[1]{\lbrace\skew{-14.7}\lbrace\!\!#1\!\!\skew{14.7}\rbrace\rbrace}

\newcommand{\bc}{boundary condition}
\newcommand{\bvp}{boundary value problem}
\newcommand{\rhs}{right-hand side}
\newcommand{\lhs}{left-hand side}

\newcommand{\WLOG}{without loss of generality}
\newcommand{\nbh}{neighborhood}
\newcommand{\IFF}{if and only if }

\newcommand{\Ga}{\alpha}

\newcommand{\Gd}{\delta}
\newcommand{\Ge}{\epsilon}

\newcommand{\Gve}{\varepsilon}

\newcommand{\Gl}{\lambda}

\newcommand{\Gth}{\theta}

\newcommand{\GD}{\Delta}

\newcommand{\GG}{\Gamma}

\newcommand{\GL}{\Lambda}

\newcommand{\GS}{\Sigma}

\newcommand{\GO}{\Omega}

\bmdefine\BGa{\alpha}
\bmdefine\BGb{\beta}
\bmdefine\BGd{\delta}
\bmdefine\BGe{\epsilon}
\bmdefine\BGve{\varepsilon}
\bmdefine\BGf{\phi}
\bmdefine\BGvf{\varphi}
\bmdefine\BGg{\gamma}
\bmdefine\BGc{\chi}
\bmdefine\BGi{\iota}
\bmdefine\BGk{\kappa}
\bmdefine\BGl{\lambda}
\bmdefine\BGn{\eta}
\bmdefine\BGm{\mu}
\bmdefine\BGv{\nu}
\bmdefine\BGp{\pi}
\bmdefine\BGth{\theta}
\bmdefine\BGvth{\vartheta}
\bmdefine\BGr{\rho}
\bmdefine\BGvr{\varrho}
\bmdefine\BGs{\sigma}
\bmdefine\BGvs{\varsigma}
\bmdefine\BGt{\tau}
\bmdefine\BGj{\tau}
\bmdefine\BGu{\upsilon}
\bmdefine\BGo{\omega}
\bmdefine\BGx{\xi}
\bmdefine\BGy{\psi}
\bmdefine\BGz{\zeta}
\bmdefine\BGD{\Delta}
\bmdefine\BGF{\Phi}
\bmdefine\BGG{\Gamma}
\bmdefine\BGL{\Lambda}
\bmdefine\BGP{\Pi}
\bmdefine\BGT{\Theta}
\bmdefine\BGS{\Sigma}
\bmdefine\BGU{\Upsilon}
\bmdefine\BGO{\Omega}
\bmdefine\BGX{\Xi}
\bmdefine\BGY{\Psi}

\bmdefine\BFM{\mathfrak{M}}
\bmdefine\BFb{\mathfrak{b}}
\bmdefine\BFk{\mathfrak{k}}
\bmdefine\BFm{\mathfrak{m}}
\bmdefine\BFu{\mathfrak{u}}
\bmdefine\BFv{\mathfrak{v}}

\newcommand{\CA}{{\mathcal A}}
\newcommand{\CB}{{\mathcal B}}

\newcommand{\CE}{{\mathcal E}}

\newcommand{\CH}{{\mathcal H}}

\newcommand{\CJ}{{\mathcal J}}

\newcommand{\CM}{{\mathcal M}}

\newcommand{\CO}{{\mathcal O}}
\newcommand{\CP}{{\mathcal P}}

\newcommand{\CS}{{\mathcal S}}

\newcommand{\CW}{{\mathcal W}}

\bmdefine\BCA{{\mathcal A}}
\bmdefine\BCB{{\mathcal B}}
\bmdefine\BCC{{\mathcal C}}
\bmdefine\BCD{{\mathcal D}}
\bmdefine\BCE{{\mathcal E}}
\bmdefine\BCF{{\mathcal F}}
\bmdefine\BCG{{\mathcal G}}
\bmdefine\BCH{{\mathcal H}}
\bmdefine\BCI{{\mathcal I}}
\bmdefine\BCJ{{\mathcal J}}
\bmdefine\BCK{{\mathcal K}}
\bmdefine\BCL{{\mathcal L}}
\bmdefine\BCM{{\mathcal M}}
\bmdefine\BCN{{\mathcal N}}
\bmdefine\BCO{{\mathcal O}}
\bmdefine\BCP{{\mathcal P}}
\bmdefine\BCQ{{\mathcal Q}}
\bmdefine\BCR{{\mathcal R}}
\bmdefine\BCS{{\mathcal S}}
\bmdefine\BCT{{\mathcal T}}
\bmdefine\BCU{{\mathcal U}}
\bmdefine\BCV{{\mathcal V}}
\bmdefine\BCW{{\mathcal W}}
\bmdefine\BCX{{\mathcal X}}
\bmdefine\BCY{{\mathcal Y}}
\bmdefine\BCZ{{\mathcal Z}}

\bmdefine\Bzr{ 0}
\bmdefine\Ba{ a}
\bmdefine\Bb{ b}
\bmdefine\Bc{ c}
\bmdefine\Bd{ d}
\bmdefine\Be{ e}
\bmdefine\Bf{ f}
\bmdefine\Bg{ g}
\bmdefine\Bh{ h}
\bmdefine\Bi{ i}
\bmdefine\Bj{ j}
\bmdefine\Bk{ k}
\bmdefine\Bl{ l}
\bmdefine\Bm{ m}
\bmdefine\Bn{ n}
\bmdefine\Bo{ o}
\bmdefine\Bp{ p}
\bmdefine\Bq{ q}
\bmdefine\Br{ r}
\bmdefine\Bs{ s}
\bmdefine\Bt{ t}
\bmdefine\Bu{ u}
\bmdefine\Bv{ v}
\bmdefine\Bw{ w}
\bmdefine\Bx{ x}
\bmdefine\By{ y}
\bmdefine\Bz{ z}
\bmdefine\BA{ A}
\bmdefine\BB{ B}
\bmdefine\BC{ C}
\bmdefine\BD{ D}
\bmdefine\BE{ E}
\bmdefine\BF{ F}
\bmdefine\BG{ G}
\bmdefine\BH{ H}
\bmdefine\BI{ I}
\bmdefine\BJ{ J}
\bmdefine\BK{ K}
\bmdefine\BL{ L}
\bmdefine\BM{ M}
\bmdefine\BN{ N}
\bmdefine\BO{ O}
\bmdefine\BP{ P}
\bmdefine\BQ{ Q}
\bmdefine\BR{ R}
\bmdefine\BS{ S}
\bmdefine\BT{ T}
\bmdefine\BU{ U}
\bmdefine\BV{ V}
\bmdefine\BW{ W}
\bmdefine\BX{ X}
\bmdefine\BY{ Y}
\bmdefine\BZ{ Z}

\newcommand{\SFC}{\mathsf{C}}

\usepackage{color}
\usepackage[normal]{subfigure}
\begin{document}
\maketitle


\begin{abstract}
Clapeyron's Theorem of   classical linear elasticity  provides a way 
 to explicitly express the energy stored in an equilibrium configuration in terms of
 the work of the forces applied on the boundary. We derive  several
 new integral relations   which can be viewed as nonlinear analogs of
 this classical result, reinterpreting them as rather general
 statements within Calculus of Variations. These relations reflect
 specific properties of Lagrangians, that we call
 ``partial variational symmetries'', since they are more general than
 classical variational symmetries. In the framework of nonlinear
 elasticity, partial variational symmetries, such as scale invariance,
 or scaling homogeneity, lead, via Noether's analysis to different
nonlinear   generalizations of  Clapeyron's Theorem that combine
 naturally the work of physical and configurational forces.  We
 present a series of illuminating examples showing the  
 effectiveness of the obtained general results in different problems of nonlinear elasticity.
\end{abstract}

\section{Introduction}
\label{sec:intro}
A well known theorem of linear elasticity, first attributed to B.P.E. Clapeyron in \cite{lame1852}, 
states that the elastic energy stored in a body loaded by (dead) forces is equal to half of the work
done by the loading device  \cite{love27,Sokolnikoff:1983:MTE}.
More precisely, it states that 
\begin{equation}
\int_{\GO}W (\Bx,e(\Bu))d\Bx=\hf\int_{\dOm}(\BGs\Bn,\Bu)dS(\Bx), \label{Clap}
\end{equation}
where
\begin{equation}
  \label{linel}
  W (\Bx,\BGve)=\hf(\SFC(\Bx)\BGve,\BGve)
\end{equation}
is the stored elastic energy density,  $\SFC(\Bx)$ is the elasticity
tensor, taking values in the space of positive definite quadratic
forms on $\Sym(\bb{R}^{n})$---the space of $n\times n$ symmetric real
matrices, and can otherwise be an arbitrary bounded and measurable
function of $\Bx$; $\Bu(\Bx)$ is the displacement vector,
$\BGs(\Bx)=\SFC(\Bx)e(\Bu)$ is the Cauchy stress tensor and
\begin{equation}
  \label{eofu}
  e(\Bu)=\hf(\Grad\Bu(\Bx)+(\Grad\Bu(\Bx))^{T})
\end{equation}
is the linear strain tensor.
While this theorem can be proved by simple integration by
parts, and is widely used in applications
\cite{brun65,maro97,huet01,momo22}, it still carries a certain flavor of mystery. 
First, it has never
been published by the author himself \cite[p.~418]{tope60}. Second, it implies  the
puzzling disappearance of the half of the work, which can be linked to the fact that
elastostatics in the problems with abrupt application of dead forces is only a weak limit of elastodynamics \cite{fotr03}. Finally,  to the best of our knowledge, Clapeyron's Theorem (CT)  \eqref{Clap}   has always been viewed as strictly limited to linear elasticity, with an implication  that a direct   analog does not exist in geometrically and physically nonlinear elasticity.  
While a seemingly related result   has been obtained  in 3D nonlinear
elasticity theory  by  J.D. Eshelby \cite {eshe70}, see also \cite
{green73,knst84,hill86}, its  relation to the classical CT remained unexplored.
 
In this paper we  reconnect  linear and nonlinear elasticity perspectives on surface representation of the bulk energy and   expose a deeper structure of  the classical CT,  unifying it with   other known and new similar integral equalities. This  will not only allows us  to generalize CT  beyond linear elasticity but will also open  the possibility to  view it as a general result in the broader Calculus of Variations. 

The crucial step in our analysis is the realization that behind  the conventional CT is a variational formula first obtained by E. Noether \cite{noether18,olver86}. We first recall  that the  well-known Noether Theorem links variational symmetries with conservation laws in the
form of divergence-free combinations of   field  variables. To obtain this result,
Noether  considered the action of a continuous group  of
symmetries of the variational functional and showed  that when a
continuous symmetry parameter is varied,  the
  corresponding variation of the functional   can be written as a sum  of the  Euler-Lagrange operator and  a divergence \cite{ibra84,olver86,blku89}. 
The implied variational   formula    has been   used broadly to generate various   integral identities  in many different  domains of science \cite{wagn02,poho65,poho70,reic04,rell56,bomi07,puse86,vdvo91,knops03,olver86}.

Following prior work, see for instance \cite{bomi07}, we move away
from symmetry groups considered by Noether towards more general
classes of transformations regarded as \emph{partial symmetries}. The
Noether formula still holds, and we use it to show that the classical
CT can be linked to degree 2 homogeneity of the
quadratic energy density of linear elasticity. As such, the classical
CT is exhibited as a member of an family of
Clapeyron-type theorems for $p$-homogeneous energy densities, for any
$p>0$. We then show that scale invariance property of nonlinear elasticity
is another instance of a partial variational symmetry, leading
to what we call the Clapeyron-Eshelby Theorem (CET), also representing the total energy as a surface integral.  
 
Remarkably, while in the classical CT the implied
boundary term can be interpreted as representing the work of physical
forces, in the CET the energy of an extremal configuration is
expressed through the work of not only physical but also of
configurational forces
\cite{eshe70,chad75,gurt00,podi02,ligu06,ssd09,maug16}. To clarify the
physical meaning of such forces we show that the work of the
configurational forces can be linked to the energetic price of elastic
incompatibility acquired during the  formation of a (externally unstressed) solid,   say via crystallization from a liquid or  a surface deposition. Given that the CET naturally
combines the work performed by the  Piola and  Eshelby stresses
\cite{eshe70,eric77,silh97,gurt99,podi01,stma05,gfa10}, we show that
there exists a higher dimensional representation of the CET featuring
a new tensor uniting the Piola stress and the Eshelby stresses into
a single tensor field. To obtain such a representation we follow Noether's geometric view
and regard the classical non-parametric variational integral as a
 functional of the surface representing the graph of the unknown
 vector-valued function, thereby converting it into a parametric
 one. We then show that the $p$-homogeneous version of CT, applied to the $n$-homogeneous parametric form of the original functional leads to a compact form involving the combined
Piola-Eshelby stress tensor. Upon returning to the original
non-parametric notation, we obtain CET, revealing close interconnections between
all Clapeyron-type theorems.

The proposed broader reading  allows us to
obtain a series of specific results from CET  which we found convenient to reformulate in terms  of the Weierstrass excess function
\cite{knst86,grtrnc,grtrmms,rosa_pre}. Also, since our work is motivated by
applications which allow for surfaces of jump discontinuities of field
gradients, in our derivation of the generalized CET we also account of the
inherent lack of smoothness of extremals
\cite{knst78,know79,ball87,grtrmms}.  In particular, this issue also
arises in nonlinear elastodynamics  where  the formation of
shocks is unavoidable   despite the presence of  the overall
variational structure \cite{dafHCL,ggks25,grtreldyn}. To highlight the
utility of the obtained results, we present a series of illuminating
examples of the effectiveness of our generalized  CET in several nontrivial problems
in nonlinear elasticity. Specifically, we use the CET to formulate a
new necessary condition of metastability, establish the non-existence
of strong local minima for a broad class of domains and boundary
condition, and gain nontrivial insight into the structure of
quasi-convex envelopes of nonconvex elastic energy densities. See also
\cite{grtrvoid,grtrpoho} for other application of several related integral representations of the energy to fracture mechanics and nonlinear PDE.
  
The paper is organized as follows.  Section~\ref{sec:noether}  contains the background material including the  derivation of Noether's formula, valid for general Lipschitz vector fields.
We apply Noether's formula to give a characterization of parametric
Lagrangians and observe that every non-parametric variational problem can
be reformulated as a parametric one. We then show how the nonlinear
versions of the  well-known $J$, $L$ and $M$ integrals in fracture
mechanics \cite{rice68} are consequences of Noether's celebrated link between
symmetries and conservation laws.
Classical CT and its various direct generalizations
are obtained in Section~\ref{sub:CCT} as a consequence of
$p$-homogeneity of the elastic energy density.  The most general
version of the Clapeyron theorem (GCT) is derived in
Section~\ref{sub:GCT1}. The CET  is obtained  in
Section~\ref{sub:GCT} where it is shown to be an outcome of scale invariance. In
Section~\ref{sub:GET_CT} we relate our  different  Clapeyron-type 
theorems together by finally subordinating CET to CT. In
Section~\ref{sub:inc} we show that CET can be generalized to account for the presence of
scale-free constraints, such as the incompressibility. We also present
in Section~\ref{sub:GCTlin} several new results for the linear elastic
energies combining $2$-homogeneity with scale invariance. 
In Section~\ref{sub:lbcf} we explain  why the work of configurational forces can be linked to the incompatibility acquired during formation of an unstressed body.
Finally, our Section~\ref{sec:app} is devoted to various nontrivial applications of
CET-type results in the general Calculus of Variations, and in nonlinear elasticity, in particular.  
The conclusions are presented in Section~\ref{sec:conc}.

\section{Calculus of variations}
\label{sec:noether}

In this section, as a preparation,  we re-derive and generalize several relevant results  in the  abstract setting  of  Calculus of Variations.

\subsection{Noether identity}
We begin by introducing  a general variational problem of minimizing the energy functional 
\begin{equation}
E[\By]=\int_{\Omega}W(\Bx,\By(\Bx),\Grad\By(\Bx))d\Bx,
\label{non-param}
\end{equation}
where, using the language of elasticity theory,  we interpret the
function $\By:\GO\to\bb{R}^{m}$ as a deformation field defined on  the
reference (Lagrangian) configuration $\GO\subset\bb{R}^{n}$; the image $\GO^{*}=\By(\GO)$ will be then the  actual or Eulerian configuration.  
We recall in this respect that in   classical   nonlinear elasticity it is assumed that the energy density  
$W(\Bx,\By,\BF)=W(\BF)$ does not depend explicitly on $\By$ and is objective, i.e.
$W(\Bx,\BR\BF)=W(\Bx,\BF)$ for all $\BR\in SO(3)$ and all $\BF$,
$\det\BF>0$. For homogeneous bodies (i.e. bodies made of a single
material) the energy is also independent of $\Bx$. 
However, in this paper we
are not making these restrictive assumptions, as many of our conclusions hold for a far more general types of problems.  In particular, with  various specific  applications in mind, we  assume the integers  $m$ and $n$ to be arbitrary.

We assume further that $\By(\Bx)$ is only Lipschitz
continuous\footnote{It is a straightforward to generalize all analyses
to functions in Sobolev classes $W^{1,p}$ by making appropriate
growth assumption on $W(\BF)$ and $W_{\BF}(\BF)$, which are not
necessary when $\By(\Bx)$ is Lipschitz.}. Here we have in mind
applications, for instance,  to martensitic phase transitions and to elastodynamics where the underlying
configurations are typically Lipschitz continuous, but not smooth
\cite{Ericksen:1980:SPT,kh,Kinderlehrer:1993:MPT}.
There are also many examples of Lipschitz extremals with more general
singular sets, even for strictly convex energy densities
\cite{degio68,necas77,ball82,sivo92,svya00,musv03}.

We will assume, in general, that the function $W(\Bx,\By,\BF)$ is of class $C^{1}$, if not globally, then
on the open set containing the range of $(\Bx,\By(\Bx),\Grad\By(\Bx))$ for any
configuration $\By$ under consideration. In the context of three-dimensional nonlinear elasticity, for example, this means that the range of $\Grad\By(\Bx)$ will be a closed subset of $\{\BF\in\bb{R}^{3\times 3}:\det\BF>0\}$, eliminating the difficulties typically arising when one considers Sobolev solutions, whose existence is guaranteed by \cite{ball7677}.
In some situations that we will identify explicitly, the regularity in the $\Bx$ variable will not be assumed, as in the case of heterogenous or composite media.

Following Noether \cite{noether18,rund72,neue17,abal23} we observe that the the graph
\[
  \GG=\{(\Bx,\By(\Bx)):\Bx\in\GO\}\subset \bb{R}^{m+n}
\]
of $\By(\Bx)$, as a geometric surface identifies the function $\By(\Bx)$
uniquely. Hence, we can regard the energy functional (\ref{non-param})
as a functional $E[\GG]$.
We then examine the effect of the graph perturbation on $E[\GG]$. To this end we introduce 
a smooth family of Lipschitz maps $\BGF_{\Ge}:\bb{R}^{m+n}\to \bb{R}^{m+n}$, such that $\BGF_{0}(\Bx,\By)=(\Bx,\By)$. When $\Ge$ is sufficiently small, $\GG_{\Ge}=\BGF_{\Ge}(\GG)$ would still be a graph of some Lipschitz function $\Tld{\By}_{\Ge}$.

Our first goal   will be  to derive the formula for the first variation
\begin{equation}
  \label{firstvar}
\Gd E=\left.\frac{dE[\GG_{\Ge}]}{d\Ge}\right|_{\Ge=0},
\end{equation}
where
\begin{equation}
  \label{firstvar2}
E[\GG_{\Ge}]=\int_{\GO_{\Ge}}W(\Tld{\Bx},\By_{\Ge}(\Tld{\Bx}),\Grad\By_{\Ge}(\Tld{\Bx}))d\Tld{\Bx}.
\end{equation}
To explain the  notations in   \eqref{firstvar2}, we write
\begin{equation}
  \label{trgrp}
\BGF_{\Ge}(\Bx,\By)=(\BX(\Bx,\By,\Ge),\BY(\Bx,\By,\Ge)),
\end{equation}
where $\BX$ and $\BY$ are Lipschitz functions of $(\Bx,\By)$ and of class $C^{1}$ in $\Ge$, such that
\[
\BX(\Bx,\By,0)=\Bx,\quad\BY(\Bx,\By,0)=\By.
\]
According to (\ref{trgrp}), $\Tld{\By}_{\Ge}(\Tld{\Bx})=\BY(\Bx,\By(\Bx),\Ge)$, in other words,
\begin{equation}
  \label{yeps}
  \By_{\Ge}(\BX(\Bx,\By(\Bx),\Ge))=\BY(\Bx,\By(\Bx),\Ge),
\end{equation}
which defines $\By_{\Ge}(\Tld{\Bx})$ in \eqref{firstvar2},  
where $\GO_{\Ge}$ is the image of $\GO$ under the transformation $\Bx\mapsto\BX(\Bx,\By(\Bx),\Ge)$.

To compute $\Gd E$ given by (\ref{firstvar}) 
we first introduce the notations 
\begin{equation}
  \label{dxdyxy}
  \left.\Gd\Bx(\Bx,\By)=\dif{\BX(\Bx,\By,\Ge)}{\Ge}\right|_{\Ge=0},\quad
\left.\Gd\By(\Bx,\By)=\dif{\BY(\Bx,\By,\Ge)}{\Ge}\right|_{\Ge=0}, 
\end{equation}
even though in our subsequent derivations we will only use functions
\begin{equation}
  \label{dxdy}
  \Gd\Bx(\Bx)=\Gd\Bx(\Bx,\By(\Bx)),\qquad \Gd\By(\Bx)=\Gd\By(\Bx,\By(\Bx)), 
\end{equation}
that are required to be at least\footnote{The exact regularity
  requirements come from the validity of the derivation of formula (\ref{deform})
for $\Gd E$.} of class $W^{1,1}$.
We then make the change of variables $\Tld{\Bx}=\BX(\Bx,\By(\Bx),\Ge)$ in \eqref{firstvar2},
and obtain
\[
E[\GG_{\Ge}]=\int_{\GO}W(\BX(\Bx,\By(\Bx),\Ge),\By_{\Ge}(\BX(\Bx,\By(\Bx),\Ge)),\Grad\By_{\Ge}(\BX(\Bx,\By(\Bx),\Ge)))\det(\Grad\BX)d\Bx,
\]
where $\Grad\BX=\Grad_{\Bx}(\BX(\Bx,\By(\Bx),\Ge))$.
Differentiating (\ref{yeps}) in $\Bx$ we obtain
\[
\Grad\By_{\Ge}(\BX(\Bx,\By(\Bx),\Ge))\Grad\BX=\Grad\BY
\]
Therefore,
\begin{equation}
  \label{deform1}
  E[\GG_{\Ge}]=\int_{\GO}W(\BX,\BY,\Grad\BY(\Grad\BX)^{-1})\det(\Grad\BX)d\Bx,
\end{equation}
where $\BX$ and $\BY$ are evaluated at $(\Bx,\By(\Bx),\Ge)$.
We note that
\[
\left.\dif{\Grad\BX}{\Ge}\right|_{\Ge=0}=\Grad\Gd\Bx(\Bx),\quad \left.\dif{\Grad\BY}{\Ge}\right|_{\Ge=0}=\Grad\Gd\By(\Bx),
\]
where $\Gd\Bx(\Bx)$ and $\Gd\By(\Bx)$ are defined in (\ref{dxdy}).
Now, we can differentiate (\ref{deform}) in $\Ge$ at $\Ge=0$, and
obtain an identity
\begin{equation}
  \label{deform}
    \Gd E=\int_{\GO}\{W_{\Bx}\cdot\Gd\Bx+W_{\By}\cdot\Gd\By+\av{W_{\BF},\Grad\Gd\By-\Grad\By\Grad\Gd\Bx}+
  W\Div\Gd\Bx\}d\Bx.
\end{equation}
We now introduce  important notations which will be  used extensively in what follows. First, we define the tensor field 
\begin{equation}
  \BP(\Bx,\By,\BF)=W_{\BF}(\Bx,\By,\BF), 
\label{Esh-tensor1}
\end{equation}
to which we  refer as  the  \emph{Piola stress tensor}, using its  narrow meaning in classical   elasticity theory; in a general field theory it is known as canonical momentum or current tensor.  Second, we define another tensor field 
\begin{equation}
  \BP^{*}(\Bx,\By,\BF)=W(\Bx,\By,\BF)\BI_{n}-\BF^{T}\BP(\Bx,\By,\BF),
\label{Esh-tensor}
\end{equation}
and refer to it as the \emph{Eshelby tensor}, again, referring to a more limited notion used in nonlinear  elasticity;  in a general field theory this object  is known as the energy-momentum tensor. The corresponding background material, explaining the origin of these definitions in more detail,  can be found in the monographs \cite{silh97,gurt00,gfa10,maug16,gori17,krro19,seep20}. 
Using these definitions we can rewrite identity \eqref{deform} in the form 
\begin{equation}
  \label{firstvar1}
  \Gd E=\int_{\GO}\{W_{\By}\cdot\Gd\By+W_{\Bx}\cdot\Gd\Bx+\av{\BP,\Grad\Gd\By}
  +\av{\BP^{*},\Grad\Gd\Bx}\}d\Bx,
\end{equation}
where $\BP$ and $\BP^{*}$ are evaluated at
$(\Bx,\By(\Bx),\Grad\By(\Bx))$. We will refer to formula
(\ref{firstvar1}) as the weak formulation of Noether's identity, as it
holds for all Lipschitz vector fields $\By(\Bx)$ and variations
$\Gd\Bx(\Bx)$ and $\Gd\By(\Bx)$ of class $W^{1,1}$. The Lagrangian
$W(\Bx,\By,\BF)$ can also be only of class $C^{1}$.

If all the functions above, $W$, $\By$, $\BX$, and $\BY$ are of class
$C^{2}$, we can integrate by parts in (\ref{firstvar1}), and rewrite it as
\begin{equation}
  \label{deriv}
\Gd E=\int_{\GO}\{\mathfrak{E}_{W}(\Bx)\cdot\Gd\By+\mathfrak{E}^{*}_{W}(\Bx)\cdot\Gd\Bx\}d\Bx+
\int_{\dOm}\{\BP\Bn\cdot\Gd\By+\BP^{*}\Bn\cdot\Gd\Bx\}d\Bx,
\end{equation}
where
\begin{equation}
  \label{EL}
  \mathfrak{E}_{W}(\Bx)=W_{\By}(\Bx,\By(\Bx),\Grad\By(\Bx))-\Div \BP(\Bx,\By(\Bx),\Grad\By),
\end{equation}
and
\begin{equation}
  \label{ELT}
\mathfrak{E}^{*}_{W}(\Bx)=W_{\Bx}(\Bx,\By(\Bx),\Grad\By(\Bx))-\Div\BP^{*}(\Bx,\By(\Bx),\Grad\By(\Bx)).
\end{equation}
Formula (\ref{deriv}) will be referred to as the strong, or classical form of Noether's identity.

As an application of formula (\ref{firstvar1}), let us show that the fields $ \mathfrak{E}_{W}(\Bx)$ and $\mathfrak{E}^{*}_{W}(\Bx)$ are linked through a relation, derived originally by  Noether \cite{noether18}:
\begin{equation}
  \label{Noether0}
\mathfrak{E}^{*}_{W}(\Bx)=-(\Grad\By)^{T}\mathfrak{E}_{W}(\Bx),
\end{equation}
for smooth vector fields $\By(\Bx)$. This smoothness requirement can be
relaxed, to a degree. 
\begin{theorem}
  \label{th:NI}
  Suppose $W$ is of class $C^{1}$, and $\bra{\By}\in W^{2,1}_{\rm loc}(\GO;\bb{R}^{m})\cap
  W^{1,\infty}(\GO;\bb{R}^{m})$. Then (\ref{Noether0})  holds in the
  sense of equality of distributions,
  where $(\Grad\By)^{T}\mathfrak{E}_{W}$ is understood as the distribution 
  \begin{equation}
    \label{FEW}
    ((\Grad\By)^{T}\mathfrak{E}_{W})_{i}=\dif{\bra{y}^{k}}{x^{i}}W_{y^{k}}-\dif{}{x^{j}}\left[\dif{\bra{y}^{k}}{x^{i}}P_{k}^{j}\right]+\mix{\bra{y}^{k}}{x^{i}}{x^{j}}P_{k}^{j}, 
  \end{equation}
while formula (\ref{ELT}) for $\mathfrak{E}^{*}_{W}$ is
interpreted in the sense of distributions.
\end{theorem}
\begin{proof}
The idea is to construct a sufficiently large family of nontrivial
variations of the $(\Bx,\By)$-space none of whose members change the graph of
$\bra{\By}(\Bx)$. Let $\BGf\in C_{0}^{\infty}(\GO;\bb{R}^{m})$ be arbitrary. Then
$\BX(\Bx,\Ge)=\Bx+\Ge\BGf(\Bx)$ is a diffeomorphism of $\GO$ onto itself for
sufficiently small $|\Ge|$, so that $\GO_{\Ge}=\GO$. Then the variation
\[
\BGF_{\Ge}(\Bx,\By)=(\BX(\Bx,\Ge),\By+\bra{\By}(\BX(\Bx,\Ge))-\bra{\By}(\Bx)),
\]
leaves the graph of $\bra{\By}(\Bx)$ invariant. Indeed,
\[
  \BGF_{\Ge}(\Bx,\bra{\By}(\Bx))=(\BX(\Bx,\Ge),\bra{\By}(\Bx)+\bra{\By}(\BX(\Bx,\Ge))
  -\bra{\By}(\Bx))=(\BX(\Bx,\Ge),\bra{\By}(\BX(\Bx,\Ge))).
\]
Since $\BX(\Bx,\Ge)$ is a diffeomorphism of $\GO$, we conclude that $\BGF_{\Ge}(\GG)=\GG$, where $\GG$ is the graph of $\bra{\By}(\Bx)$. Thus, $\Gd E=0$ on the \lhs\ of (\ref{firstvar1}). We also compute
\[
  \Gd\Bx=\BGf,\quad\Gd\By=(\Grad\bra{\By})\BGf\in W_{0}^{1,1}(\GO;\bb{R}^{m}).
\]
In that case formula (\ref{firstvar1}) reads
\begin{equation}
  \label{preNI}
  0=\int_{\GO}\{W_{\Bx}\cdot\BGf+\av{\BP^{*},\Grad\BGf}+W_{\By}\cdot(\Grad\bra{\By})\BGf+
  \av{\BP,\Grad((\Grad\bra{\By})\BGf)}\}d\Bx,
\end{equation}
writing
\[
  \av{\BP,\Grad((\Grad\bra{\By})\BGf)}=P_{k}^{j}\dif{\bra{y}^{k}}{x^{i}}\dif{\phi^{i}}{x^{j}}
  +P_{k}^{j}\mix{\bra{y}^{k}}{x^{i}}{x^{j}}\phi^{i},
\]
we see that (\ref{preNI}) can be written as
$0=\av{\mathfrak{E}^{*}_{W}|\BGf}+\av{(\Grad\By)^{T}\mathfrak{E}_{W}|\BGf}$,
where the distribution $(\Grad\By)^{T}\mathfrak{E}_{W}$ is defined by (\ref{FEW}),
and where angular brackets with a bar denote the action of distributions
on test functions.
\end{proof}
\begin{remark} \label{rem1}
 When $\bra{\By}$ is merely Lipschitz continuous, not only the proof of
  Theorem~\ref{th:NI} breaks down, but formula
(\ref{firstvar1}) itself is no longer valid. 
Nonetheless, Theorem~\ref{th:NI} can still be used on
subdomains of $\GO$, where $\bra{\By}$ is of class $W^{2,1}$. 
\end{remark}

\subsection{Parametric representation}

An important  application of formula (\ref{firstvar1}) concerns  parametric variational integrals
\begin{equation}
  \label{Egraph1}
E[\GS]=\int_{D}W(\Bz(\Bt),\Grad\Bz(\Bt))d\Bt,
\end{equation}
which do not depend on the choice of parametrization of a Lipschitz $n$-dimensional
surface \[\GS=\{\Bz(\Bt):\Bt\in D\}\subset\bb{R}^{m}.\] To reformulate independence of
reparametrization in terms of the integrand $W(\Bz,\BF)$ we consider
the reparametrization
$\Bs=\Bt+\Ge\BGth(\Bt)$, where $\BGth:D\to\bb{R}^{n}$, is an arbitrary
smooth function and $|\Ge|$ is sufficiently small. The new parametrization is then
$\Bz_{\Ge}(\Bs)$, satisfying
\[
\Bz(\Bt)=\Bz_{\Ge}(\Bs)=\Bz_{\Ge}(\Bt+\Ge\BGth(\Bt)),\quad\Bt\in D.
\]
By definition of the parametric integral we must have
\begin{equation}
  \label{rep}
\int_{D}W(\Bz(\Bt),\Grad\Bz(\Bt))d\Bt=\int_{D_{\Ge}}W(\Bz_{\Ge}(\Bs),\Grad\Bz_{\Ge}(\Bs))d\Bs,
\end{equation}
where $D_{\Ge}=\{\Bt+\Ge\BGth(\Bt):\Bt\in D\}$.
But then, formula (\ref{firstvar1}) applied to the transformation
\[
\BGF_{\Ge}(\Bt,\Bz)=(\Bt+\Ge\BGth(\Bt),\Bz),
\]
yields, taking into account that $\Gd\Bt=\BGth(\Bt)$, $\Gd\Bz=0$,
\[
  0=\int_{D}\av{\BP^{*},\Grad\BGth}d\Bt.
\]
Since, both the domain $D$ and the vector field $\BGth$ were
arbitrary, we conclude that
\begin{equation}
  \label{Pst0}
  \BP^{*}(\Bz,\BF)=0\quad\eqv\quad\BF^{T}W_{\BF}(\Bz,\BF)=W (\Bz,\BF)\BI_{n}
\end{equation}
identically. Equation  \eqref{Pst0} is the structural equation
satisfied by any parametric integrand $W(\Bz,\BF)$. Equation
(\ref{Pst0}) is equivalent to the property
\begin{equation}
  \label{vechom}
  W(\Bz,\BF\BA)=W(\Bz,\BF)\det\BA
\end{equation}
for any $\BA\in GL^{+}(n)$.
Indeed, if (\ref{vechom}) holds for all $\BA\in GL^{+}(n)$, then,
differentiating (\ref{vechom}) in $\BA$ at $\BA=\BI_{n}$, we obtain
$\BP^{*}(\Bz,\BF)=0$. Conversely, if  (\ref{Pst0}) holds, then differentiating the function
\[
GL^{+}(n)\ni\BA\mapsto\Phi(\BA)=\frac{W(\Bz,\BF\BA)}{\det\BA}
\]
for arbitrary fixed values of $\Bz\in\bb{R}^{m}$ and $\BF\in\bb{R}^{m\times n}$,
we obtain
\[
  \Grad\Phi(\BA)=\frac{\BF^{T}W_{\BF}(\Bz,\BF\BA)-W(\Bz,\BF\BA)\BA^{-T}}{\det\BA}=
  \frac{\BA^{-T}}{\det\BA}\BP^{*}(\Bz,\BF\BA)=0.
\]
Hence, $\Phi(\BA)=\Phi(\BI_{n})$ for all $\BA\in GL^{+}(n)$, since
$GL^{+}(n)$ is connected. Relation (\ref{vechom}) follows.
We conclude that if $W(\Bz,\BF)$ is a parametric Lagrangian, then it
must have the property (\ref{vechom}), which is equivalent to (\ref{Pst0}).
Conversely, if (\ref{vechom}) is satisfied and $\Bt=\BT(\Bs)$, $\Bs\in D'$ 
is an orientation-preserving Lipschitz bijection, then $\Tld{\Bz}(\Bs)=\Bz(\BT(\Bs))$ is
another parametrization of the surface $\GS$. Hence,
\begin{multline*}
    \int_{D'}W(\Tld{\Bz}(\Bs),\Grad \Tld{\Bz}(\Bs))d\Bs=
    \int_{D'}W(\Bz(\BT(\Bs)),\Grad\Bz(\BT(\Bs))\Grad\BT(\Bs))d\Bs=\\
    \int_{D'}W(\Bz(\BT(\Bs)),\Grad\Bz(\BT(\Bs)))\det\Grad\BT(\Bs)\,d\Bs=
    \int_{D}W(\Bz(\Bt),\Grad\Bz(\Bt))d\Bt,
  \end{multline*}
  since $\Grad\BT(\Bs)\in GL^{+}(n)$ for all $\Bs\in D'$.
Therefore,  the algebraic property (\ref{vechom}), or (\ref{Pst0})
gives a complete characterization of parametric Lagrangians. A well known  example of parametric invariance in the context of nonlinear elasticity is the  indifference to ``material remodeling", see for instance \cite{epel07}.

We now show  that any non-parametric variational integral
(\ref{non-param}) can be written as a parametric integral by making an
arbitrary smooth change of variables 
\begin{equation}
  \label{graphu1}
  \Bx=\BGn(\Bt)
  \end{equation}
in (\ref{non-param}), where $\BGn:\bra{D}\to\bra{\GO}$ is a
diffeomorphism. We obtain
\[
E[\By]=\int_{D}W(\BGn(\Bt),\By(\BGn(\Bt)),\Grad\By(\BGn(\Bt)))\det\Grad\BGn(\Bt)d\Bt.
\]
Introducing the notation
\[
  \Bz(\Bt)=(\BGn(\Bt),\By(\BGn(\Bt)))=(\Bz_{1}(\Bt),\Bz_{2}(t))
\]
we can now write
\begin{equation}
  \label{param}
    E[\By]=\int_{D}W(\Bz(\Bt),\Grad\Bz_{2}(\Bt)(\Grad\Bz_{1}(\Bt))^{-1})\det(\Grad\Bz_{1}(\Bt))d\Bt:=
  \int_{D}\CW(\Bz(\Bt),\Grad\Bz(\Bt))d\Bt, 
\end{equation}
Thus, the original non-parametric Lagrangian $W(\Bx,\By,\BF)$ has been replaced by the
 parametric Lagrangian operating in the extended space 
 \begin{equation}
   \label{grL}
   \CW(\Bz,\BCF)=W(\Bz_{1},\Bz_{2},\BF_{2}\BF_{1}^{-1})\det\BF_{1}, 
 \end{equation}
where $\Bz=(\Bz_{1},\Bz_{2})\in\bb{R}^{n+m}$, $\BCF=[\BF_{1};\BF_{2}]\in\bb{R}^{(m+n)\times n}$, where
$\BF_{1}$ is the upper $n\times n$ submatrix of $\BCF$, and $\BF_{2}$
is the lower $m\times n$ submatrix of $\BCF$. The function
$W(\Bz,\BCF)$ we have constructed satisfies (\ref{vechom}). Indeed, if
$\BA\in GL^{+}(n)$, then $\BCF\BA=[\BF_{1}\BA;\BF_{2}\BA]$, and hence,
\[
\CW(\Bz,\BCF\BA)=W(\Bz_{1},\Bz_{2},\BF_{2}\BA(\BF_{1}\BA)^{-1})\det\BF_{1}\BA=\CW(\Bz,\BCF)\det\BA.
\]
Therefore, equation (\ref{param}) gives a  non-parametric variational integral
(\ref{non-param}) in terms of a parametric Lagrangian $\CW(\Bz,\BCF)$.
Note that even if the original Lagrangian $W(\Bx,\By,\BF)$ was smooth on
$\bb{R}^{n}\times\bb{R}^{m}\times\bb{R}^{m\times n}$, the new Lagrangian is
smooth only on the open subset $\CO$ of $\bb{R}^{n+m}\times\bb{R}^{(m+n)\times n}$,
on which $\BF_{1}$ is invertible.

\subsection{Noether formula}
Since our  primary interest lies in studying energy minimizers
 we now consider the simplifications in the Noether's formula induced by the optimality conditions 
\begin{equation}
\Gd E=0  \label{ELeq1}
 \end{equation}
 for all outer variations $\Gd\By\in C_{0}^{\infty}(\GO;\bb{R}^{m})$  and
all inner variations $\Gd\Bx\in
C_{0}^{\infty}(\GO;\bb{R}^{n})$. According to formula
(\ref{firstvar1}), the condition \eqref{ELeq1}   requires that 
\begin{equation}
    \label{ELeq}
\mathfrak{E}_{W}(\Bx)=0,
\end{equation}
and
\begin{equation}
  \label{stationary}
\mathfrak{E}^{*}_{W}(\Bx)=0,
\end{equation}
understood in the sense of distributions. Theorem~\ref{th:NI} shows
that smooth configurations $\By(\Bx)$ satisfying (\ref{ELeq}) will
automatically satisfy (\ref{stationary}). Hence, in classical theories
equation (\ref{ELeq}) serves as the main governing equation. In
general, a Lipschitz configuration $\By(\Bx)$ satisfying (\ref{ELeq})
need not solve (\ref{stationary}).
\begin{definition}
  \label{def:equil}
We say that the Lipschitz configuration $\By(\Bx)$ is an extremal, if it satisfies the Euler-Lagrange equation (\ref{ELeq}) in the sense of distributions.
\end{definition}
An extremal is a critical point of the functional (\ref{non-param}) in
the $W^{1,\infty}$ norm topology. The so-called ``least action
principle'' requires that the dynamics of a mechanical system be an
extremal of the action functional, without the requirement that the
action be minimized. Remark~\ref{rem1} implies that the distribution
$\mathfrak{E}^{*}_{W}$ must be supported on the singular sets of
space-time gradients of extremals, while thermodynamical
considerations are invoked to further constrain the distribution
$\mathfrak{E}^{*}_{W}$, \cite{trus93,grtreldyn}.

If the energy functional (\ref{non-param}) is to be minimized, then
(\ref{stationary}) must hold, in addition to (\ref{ELeq}), for all
Lipschitz energy minimizers.
\begin{definition}
  \label{def:exst}
We say that a Lipschitz configuration $\By(\Bx)$ is \emph{stationary} if
it is an extremal in the sense of Definition~\ref{def:equil}, and additionally satisfies (\ref{stationary}).
\end{definition}
We stress that stationarity, is not a consequence of equilibrium
for nonsmooth $\By(\Bx)$. However, any Lipschitz extremal of parametric
variational integrals is always stationary, due to (\ref{Pst0}). In
particular, any extremal of (\ref{param}) must be a stationary
extremal for (\ref{non-param}). Intuitively, the statement is apparent, since any weak outer
variation of $\Bz=(\Bx,\By)$ is a combination of an inner and outer variations
for $\By(\Bx)$, and it can also be verified by a direct calculation.
Indeed, by direct differentiation we obtain
\begin{equation}
  \label{Phat12}
\BCP(\Bz,\BCF)=\CW_{\BCF}=\left[\begin{array}{c}
\BP^{*}(\Bz,\BF_{2}\BF_{1}^{-1})\cof\BF_{1}\\[1ex]
\BP(\Bz,\BF_{2}\BF_{1}^{-1})\cof\BF_{1}
\end{array} \right].
\end{equation}
Note that such  a higher dimensional  Piola stress $\BCP$ mixes the
classical low dimensional Piola and Eshelby stresses.
The extended Euler-Lagrange equation
\[
  \Grad_{\Bt}\cdot\CW_{\BCF}(\Bz,\BCF)=\CW_{\Bz}
\]
can be written in terms of $\By(\Bx)$ as a system
\begin{equation}
  \label{ELt1}
  \begin{cases}
  \Grad_{\Bt}\cdot\left[\BP^{*}(\BGn(\Bt),\By(\BGn(\Bt)),\Grad_{\Bx}\By(\BGn(\Bt)))
\cof\Grad_{\Bt}\BGn(\Bt)\right]=W_{\Bx}\det\Grad_{\Bt}\BGn(\Bt),\\
\Grad_{\Bt}\cdot(\BP(\BGn(\Bt),\By(\BGn(\Bt)),\Grad_{\Bx}\By(\BGn(\Bt)))
\cof\Grad_{\Bt}\BGn(\Bt))=W_{\By}\det\Grad_{\Bt}\BGn(\Bt).
\end{cases}
\end{equation}
It is then clear  that a smooth change of variables $\Bx=\BGn(\Bt)$ converts
 system  (\ref{ELt1}) into an equivalent system
 (\ref{ELeq}), (\ref{stationary}).
 For stationary Lipschitz extremals formula (\ref{firstvar1}) simplifies, even
 if $\Gd\Bx(\Bx)$ and $\Gd\By(\Bx)$ are only of class $W^{1,1}$.
\begin{theorem}[Noether formula]~
  \label{th:genoether}
 Suppose $\By(\Bx)$ is a Lipschitz stationary extremal
  of the functional (\ref{non-param}). Suppose also that the
  variation (\ref{trgrp}) are such that $\Gd\Bx(\Bx)$ and
  $\Gd\By(\Bx)$, given by (\ref{dxdy}) are of class $W^{1,1}$. Then
    \begin{equation}
  \label{increm}
  \Gd E=\int_{\dOm}\{\BP\Bn\cdot\Gd\By+\BP^{*}\Bn\cdot\Gd\Bx\}dS(\Bx),
\end{equation}
where the \rhs\ is understood in the sense of traces\footnote{It means that
there exists $\Bt\in L^{\infty}(\dOm;\bb{R}^{m})$, denoted by $\BP\Bn$, such that
$\displaystyle\int_{\GO}\av{\BP,\Grad\BGf}d\Bx=-\int_{\GO}\BGf\cdot\Div\BP\,d\Bx+\int_{\dOm}\Bt(\Bx)\cdot\BGf(\Bx)dS$
for any $\BGf\in C^{\infty}(\bra{\GO};\bb{R}^{m})$. The existence of
such a trace $\Bt$ follows from the assumption that
$\Div\BP=W_{\By}\in L^{\infty}(\GO;\bb{R}^{m})$.}.
\end{theorem}
\begin{proof}
The idea is to show that under the assumptions of the theorem the \rhs\ of (\ref{firstvar1}) is equal to the \rhs\ of (\ref{increm}). By density theorems there
  exist sequences $\BGf_{n}\in C^{\infty}(\bra{\GO};\bb{R}^{n})$ and
  $\BGy_{n}\in C^{\infty}(\bra{\GO};\bb{R}^{m})$, such that
  $\BGf_{n}\to\Gd\Bx$ in $W^{1,1}(\GO;\bb{R}^{n})$ and
  $\BGy_{n}\to\Gd\By$ in $W^{1,1}(\GO;\bb{R}^{m})$.
  Then, by definition of traces we have
  \[
\int_{\GO}\av{\BP,\Grad\BGy_{n}}d\Bx=-\int_{\GO}\BGy_{n}\cdot\Div\BP\,d\Bx+\int_{\dOm}\BP\Bn(\Bx)\cdot\BGy_{n}(\Bx)dS   
\]
and
 \[
\int_{\GO}\av{\BP^{*},\Grad\BGf_{n}}d\Bx=-\int_{\GO}\BGf_{n}\cdot\Div\BP^{*}\,d\Bx+\int_{\dOm}\BP^{*}\Bn(\Bx)\cdot\BGf_{n}(\Bx)dS.
\]
Since $\By$ is a stationary extremal we conclude that
\[
\int_{\GO}\{W_{\By}\cdot\BGy_{n}+W_{\Bx}\cdot\BGf_{n}+\av{\BP,\Grad\BGy_{n}}
  +\av{\BP^{*},\Grad\BGf_{n}}\}d\Bx=\int_{\dOm}\{\BP\Bn\cdot\BGy_{n}+\BP^{*}\Bn\cdot\BGf_{n}\}d\Bx.
\]
Passing to the limit as $n\to\infty$ in the formula above,
using trace inequality (e.g. \cite[Lemma~9.9]{brez11}),
shows that
the \rhs\ of (\ref{firstvar1}) is equal to the \rhs\ of (\ref{increm}).
\end{proof}
It is important to note that one can  rewrite (\ref{increm}) in a form that makes it clear that only the part of the virtual graph-displacement $(\Gd\Bx,\Gd\By)$ that moves the graph of $\By(\Bx)$ as a geometric object contributes to the integral on the \rhs\ of (\ref{increm}). Indeed,
we observe that
\[
T_{(\Bx,\By(\Bx))}\Md\GG=\{(\BGt,\Grad\By(\Bx)\BGt):\BGt\in T_{\Bx}\dOm\}
\]
is the tangent space to the boundary of the graph of $\By(\Bx)$ at the point $(\Bx,\By(\Bx))\in\Md\GG$, $\Bx\in\dOm$. We then compute, using formula (\ref{Esh-tensor}),
\[
(\BP^{*}\Bn,\BP\Bn)\cdot(\BGt,\BF\BGt)=\BP^{*}\Bn\cdot\BGt+\BP\Bn\cdot\BF\BGt=W\Bn\cdot\BGt-\BF^{T}\BP\Bn\cdot\BGt+\BP\Bn\cdot \BF\BGt=0,
\]
where we have used the shorthand $\BF$ to denote $\Grad\By(\Bx)$, and where $\BP$ and $\BP^{*}$ are evaluated at $(\Bx,\By(\Bx),\Grad\By(\Bx))$.
This shows that the vector $(\BP^{*}\Bn,\BP\Bn)\in\bb{R}^{m+n}$ is
orthogonal to $\Md\GG$ at every point of $\Md\GG$. Therefore, we can also write
\begin{equation}
  \label{grincrgeom}
  \Gd E=\int_{\dOm}\{(\BP^{*}\Bn,\BP\Bn)\cdot\CP_{(\Md\GG)^{\perp}}(\Bx)(\Gd\Bx,\Gd\By)\}dS,
\end{equation}
where $\CP_{(\Md\GG)^{\perp}}(\Bx)$ denotes the orthogonal projection
from $\bb{R}^{m+n}$ onto $(T_{(\Bx,\By(\Bx))}\Md\GG)^{\perp}$. The
above calculations are valid whenever $\Grad\By(\Bx)$ is sufficiently
regular in a \nbh\ of $\dOm$, so that $\BP\Bn$ and $\BP^{*}\Bn$ can be
interpreted pointwise on $\dOm$.

\subsection{Conservation laws}
\label{sub:varsym}
The main impact of Noether formula (\ref{increm}) was in its
application to cases of variational symmetries---transformations (\ref{trgrp}), that do not
change the functional $E[\By]$, given by (\ref{non-param}). Here we
revisit those that have applications to fracture mechanics and
elasticity more generally.
Suppose first, that the function $W$ does not
depend explicitly on $\Bx$. Therefore, the transformation
(\ref{trgrp}), given by
\begin{equation}
  \label{translation}
  \BGF_{\Ge}(\Bx,\By)=(\Bx+\Ge\Ba,\By).
\end{equation}
is a variational symmetry. Then, according to Theorem~\ref{th:genoether},
we must have
\begin{equation}
  \label{Jreason}
  \int_{\Md V}\BP^{*}\Bn\,dS=0
\end{equation}
for any piecewise smooth subdomain $V\subset\GO$.
To show  that such  statement can be useful in applications,  we can
recall the concept of  J-integral in fracture mechanics, which  in two
space dimensions is defined as a vector with two components
\begin{equation}
  \label{J1J2}
  J_{1}=\int_{C}\left\{Wdx_{2}-\Bt\cdot\dif{\By}{x_{1}}dl\right\},\qquad
J_{2}=-\int_{C}\left\{Wdx_{1}+\Bt\cdot\dif{\By}{x_{2}}dl\right\},
\end{equation}
obtained from (\ref{Jreason}) by expanding $\BP^{*}$ via (\ref{Esh-tensor})
and using  $\Bn\,dl=(dx_{2},-dx_{1})$, and $\Bt=\BP\Bn$.

While in  \eqref{J1J2} the contour 
$C$ can represent  any  path in $\GO$, in fracture mechanics one usually considers a curve $C$ that
starts at one side of the straight crack along the $x_{1}$-axis, goes around its tip,
and ends at the other side of the crack on the $x_{1}$-axis. If the crack sides are traction-free
($\Bt=0$ along the crack), then it is easy to see from (\ref{J1J2}) that $J_{1}$ (the
J-integral) does not depend on the specific locations of the start and end-points
of $C$ at the sides of the crack, and therefore gives in particular a   measure of stress
concentration at the crack tip. Note, however,  that the homotopy
invariance of the integral 
\begin{equation}
  \label{Jreason1}
\BJ=\int_{C}\BP^{*}\Bn\,dS 
\end{equation}
is shown here under more general assumptions
than it is usually done in fracture mechanics  \cite{rice68,buri73,chsh77,delph82,cher89}
in that we permit the energy density to depend on both $\By$ and $\BF$.

Another example of a symmetry and the corresponding conservation law occurs if we replace translational invariance by rotational invariance. In mechanics the corresponding  homotopy-invariant quantity, called the L-integral in \cite{buri73},  was introduced in the context of geometrically linear elasticity that does not
make a distinction between a reference and a deformed configuration.
To give a variational interpretation  of the L-integral in the
geometrically nonlinear setting, we define rotational invariance in a way that is not quite
natural from the point of view of finite elasticity.
\begin{definition}
  \label{def:iso}
  We say that the energy density $W(\Bx,\By,\BF)$ is rotationally-invariant if 
  \begin{equation}
    \label{gliniso}
    W(\BR\Bx,\BR\By,\BR\BF\BR^{T})=W(\Bx,\By,\BF),\quad\forall\BR\in SO(3).
  \end{equation}
Note, in particular,  that homogeneous nonlinearly elastic materials\footnote{Materials
whose energy density $W(\BF)$, depends only on $\BF$, and satisfies $W(\BR\BF)=W(\BF)$ for all
$\BR\in SO(3)$.} are isotropic \IFF they are rotationally-invariant .
\end{definition}

It is easy to check that the transformation group (\ref{trgrp}), given by
\[
\BGF_{\Ge}(\Bx,\By)=(e^{\Ge\BGO(\BGo)}\Bx,e^{\Ge\BGO(\BGo)}\By),
\]
where $\BGO(\BGo)$ is the skew-symmetric $3\times 3$ matrix satisfying $\BGO(\BGo)\Bu=\BGo\times\Bu$,
is a variational symmetry, since $\BR(\BGo)=e^{\Ge\BGO(\BGo)}\in SO(3)$ for all $\Ge\in\bb{R}$ and all
$\BGo\in\bb{R}^{3}$.
It is clear that the \lhs\ in Noether formula
(\ref{increm}) is zero and that $\Gd\Bx=\BGo\times\Bx$ and 
$\Gd\By=\BGo\times\By$. Taking into account that $\BGo\in\bb{R}^{3}$
is arbitrary,  we obtain the corresponding conservation law
\begin{equation}
  \label{Lreason}
  \int_{\Md V}\{\BP\Bn\times\By+\BP^{*}\Bn\times\Bx\}dS=0,
\end{equation}
 where $V\subset\GO$ is any piecewise smooth subdomain. Equation  \eqref{Lreason}  suggests that 
\[
\BL=\int_{\GS}\{\BP\Bn\times\By+\BP^{*}\Bn\times\Bx\}dS
\]
where $\GS$ is any piecewise-smooth surface, is homotopy-invariant,
provided $\Md\GS$ is fixed. If $\GS$ has no boundary, then
the nonzero values of $\BL$ can come
from closed surfaces $\GS$ containing one or more components of the crack
set in its interior.
We point out again that the homotopy
invariance of the L-integral is proved here under more general assumptions
than, for instance,  in the classical paper \cite{buri73}, in that we permit the energy density to depend (smoothly) on $\Bx$,
$\By$, and $\BF$, as long as the isotropy property (\ref{gliniso}) holds.

To obtain yet another nontrivial  invariant integral we recall that in linear elasticity 
the corresponding energy density is a quadratic function of the displacement gradient.
The emerging scaling invariance  property can then be formulated as a
variational symmetry.
More generally, we assume, that
\begin{equation}
  \label{scaling}
  W(\Gl^{-\frac{n}{p}}\Bx,\Gl^{1-\frac{n}{p}}\By,\Gl\BF)=\Gl^{p}W(\Bx,\By,\BF),\quad
  \forall\Bx\in\GO,\ \By\not=0,\ \BF\not=0,
\end{equation}
for some $p\not=0$. Here $\Gl>0$ can be restricted to an interval,
containing $\Gl=1$. 
The corresponding transformation group  
\[
\BGF_{\Ge}(\Bx,\By)=(e^{\Ge}\Bx,e^{\frac{p-n}{p}\Ge}\By)
\]
is a variational symmetry with infinitesimal generators $\Gd\Bx=\Bx$ and
$\Gd\By=\frac{p-n}{p}\By$. Theorem~\ref{th:genoether} then gives the
following corresponding conservation law
\begin{equation}
  \label{Mreason}
  \int_{\Md V}\left\{\BP^{*}\Bn\cdot\Bx+\frac{p-n}{p}\BP\Bn\cdot\By\right\}dS=0.
\end{equation}
In the special case $p=2$, $n=2$, \eqref{Mreason} reduces to the  classical  M-integral which we write in the form
\[
M_{2D}=\int_{\GS}\BP^{*}\Bn\cdot\Bx\,dS,
\]
If  $p=2$ and $n=3$, we obtain a more complex expression 
\[
M_{3D}=\int_{\GS}\left\{\BP^{*}\Bn\cdot\Bx-\hf\BP\Bn\cdot\By\right\}dS,
\]
which can be viewed as a generalization of the relation first obtained
in \cite[formula~(3b)]{buri73}. For more information about path independent integrals in elastostatics and elastodynamics  and their relation to Noether's classical paper, see \cite{flet76,herr82,delph82,haol98,mark06,luma07,ngkp24}.

\section{Clapeyron-type theorems}
\label{sec:clapeyron}
The main significance of the  general formalism developed in the previous section lies in the ability
of Theorem~\ref{th:genoether} to convert special properties of
variational functionals, which we refer to as \emph{partial
  variational symmetries}, into useful formulas satisfied by their
stationary configurations via the Noether formula (\ref{increm}).
The idea is to use particular variations $\BGF_{\Ge}$ of the $(\Bx,\By)$
space under which the Lagrangian $W(\Bx,\By,\BF)$ transforms in a known way due to its special
properties with respect to the transformation $\BGF_{\Ge}$. The original
application of this idea by Noether was to the case of variational
symmetries that we have discussed in Section~\ref{sub:varsym}.

\subsection{Clapeyron's theorem (CT)}
\label{sub:CCT}
We show now that the classical CT (\ref{Clap}) is a
particular case of Noether formula (\ref{increm}) that exploits
the partial symmetry of 2-homogeneity of the energy density function in linear
elasticity. More generally, we  now assume that the Lagrangian $W$ has the following property of
$p$-homogeneity
\begin{equation}
  \label{phom}
  W(\Bx,\Gl\By,\Gl\BF)=\Gl^{p}W(\Bx,\By,\BF),\quad\forall\Bx\in\GO,\ \By\not=0,\ \BF\not=0,
\end{equation}
for some $p\not=0$. Property (\ref{phom}) is more natural than
(\ref{scaling}), and admits heterogeneous structures made of different
$p$-homogeneous materials. 
In this case the variation
\begin{equation}
  \label{uscale}
  \Tld{\Bx}=\Bx,\qquad\Tld{\By}=e^{\Ge}\By
\end{equation}
results in $\Tld{\By}(\Tld{\Bx})=e^{\Ge}\By(\Bx)=e^{\Ge}\By(\Tld{\Bx})$, so that
\[
E[\Tld{\By}](\Ge)=\int_{\GO}W(\Tld{\Bx},e^{\Ge}\By(\Tld{\Bx}),e^{\Ge}\Grad\By(\Tld{\Bx}))d\Tld{\Bx}=
e^{p\Ge}E[\By].
\]
We have now  $\Gd\Bx=0$, $\Gd\By=\By$  which leads to our first
generalization of Clapeyron's theorem.
\begin{theorem}
  \label{th:pGCT}
  Suppose that $\By(\Bx)$ is a Lipschitz extremal of
  (\ref{non-param}). If the energy density function satisfies
  (\ref{phom}), then
  \begin{equation}
  \label{semiC}
  E[\By]=\nth{p}\int_{\dOm}\BP\Bn\cdot\By\,dS. 
\end{equation}
in the sense of traces.
\end{theorem}
We remark that in nonlinear elasticity the deformation $\By(\Bx)$ is
an orientation preserving diffeomorphism mapping the reference
configuration $\GO$ into the deformed configuration $\GO^{*}$. We
recall that in this context the Cauchy stress tensor $\BGs(\By)$ is
related to the Piola stress tensor $\BP(\Bx)$ by
\[
\BGs(\By(\Bx))\BN\,dS^{*}=\BP(\Bx)\Bn\,dS,
\]
where $\BN$ is the outward unit normal to $\dOm^{*}$ at $\By(\Bx)$,
while $\Bn$ is the outward unit normal to $\dOm$ at
$\Bx\in\dOm$. Hence, in application to nonlinear elasticity, the
$p$-homogeneous Clapeyron theorem can also be written
 entirely in terms of Eulerian coordinates
\begin{equation}
  \label{phomClap}
E[\By]=\nth{p}\int_{\dOm^{*}}\BGs(\By)\BN\cdot\By\,dS^{*}.   
\end{equation}
This means that $p$-homogeneity of the energy density ensures static
determinacy, where the total elastic energy depends only on the current state
of stress and not on the underlying displacements or
material moduli.

The classical CT (\ref{Clap}) is obtained from
Theorem~\ref{th:pGCT}, applied to the linear elastic energy density
(\ref{linel}). The absense of inner variations $\Gd\Bx=0$, according to (\ref{uscale}), makes
formula (\ref{semiC}) valid for all extremals, including the ones that
are not stationary in the sense of Definition~\ref{def:exst}.
  Also, since $\Gd\Bx=0$  the relation
(\ref{semiC}) remains  valid even when the Lagrangian $W(\Bx,\By,\BF)$ as a function of $\Bx$ is only
 measurable and bounded.

 \subsection{ Generalized Clapeyron Theorem  (GCT)}
 \label{sub:GCT1}
We are now in the position to state  a  simple but general  result to which we can  refer
  as the Generalized Clapeyron Theorem (GCT) because it  holds for any
 Lipschitz stationary extremal of an  arbitrary energy functional of the form 
 (\ref{non-param}). The corresponding theorem  will be the source of all subsequent
 results.
 \begin{theorem}
   \label{th:GCT}
   Suppose $\By(\Bx)$ is a Lipschitz stationary extremal of the energy functional
 (\ref{non-param}). Then
  \begin{equation}
   \label{genClap}
  E[\By]=-\frac{1}{n}\int_{\GO}\{W_{\Bx}\cdot\Bx+W_{\By}\cdot\By(\Bx)\}d\Bx
   +\frac{1}{n}\int_{\dOm}\{\BP\Bn\cdot\By+\BP^{*}\Bn\cdot\Bx\}dS,
 \end{equation}
 where $\BP\Bn$ and $\BP^{*}\Bn$ in the surface integral are
 understood in the sense of traces, as in the Noether formula
 (\ref{increm}). The term $-W_{\By}$ in the volume integral term has the meaning of externally applied bulk forces, while $-W_{\Bx}$ describes external configurational bulk forces.    
\end{theorem}
 \begin{proof}
   Consider a  class of scaling transformations of the form 
 \begin{equation}
   \label{scalingxy}
   \Tld{\Bx}=e^{\Ge}\Bx,\quad \Tld{\By}=e^{\Ge}\By,  
 \end{equation}
where  $\Gd\Bx=\Bx$, $\Gd\By=\By$, and
$\Tld{\By}(\Tld{\Bx})=e^{\Ge}\By(\Bx)=e^{\Ge}\By(e^{-\Ge}\Tld{\Bx})$. For
our general class of functionals (\ref{non-param}) we compute
 \begin{equation}
   \label{deltaE}
     \Gd  E=\left.\frac{d}{d\Ge}
     \int_{\GO}W(e^{\Ge}\Bx,e^{\Ge}\By(\Bx),\Grad\By(\Bx))e^{n\Ge}d\Bx\right|_{\Ge=0}=
   n E[\By]+\int_{\GO}\{W_{\Bx}\cdot\Bx+W_{\By}\cdot\By(\Bx)\}d\Bx. 
 \end{equation}
 Since $\By(\Bx)$ is a Lipschitz stationary extremal, all
 assumptions of Theorem~\ref{th:genoether} are satisfied, and 
 therefore, the generalized Noether's formula \eqref{increm} is valid.
 Then, the representation (\ref{genClap}) of the elastic energy follows.
 \end{proof}

\subsection{ Clapeyron-Eshelby theorem (CET)}
 \label{sub:GCT}
  Next we show  how  a particular partial variational symmetry,
  representing scale invariance of an associated variational problem,
  produces a simpler formula  for the stored elastic energy of the body
  in stable equilibrium in terms of physical and configurational
  tractions on its boundary.  
By scale invariance of the variational problem we mean that 
\begin{equation}
  \label{scalefree}
  W(\Gl\Bx,\Gl\By,\BF)=W(\Bx,\By,\BF),\quad\forall\Gl>0,
\end{equation}
which says that the simultaneous change of scale of Lagrangian and
Eulerian coordinates does not affect the energy density. We note that
if $W$ is continuous on $\bb{R}^{n}\times\bb{R}^{m}\times\bb{R}^{m\times n}$, property (\ref{scalefree})
simply means that $W=W(\BF)$ and does not depend on $\Bx$ and $\By$
explicitly. If $W$ is not required to be continuous at $\Bx=0$ or
$\By=0$, then the class of functions satisfying (\ref{scalefree}) is
much broader.

The assumption (\ref{scalefree}) says that $W$ is homogeneous of
degree zero as a function of $(\Bx,\By)$, and therefore, by Euler's
formula for $p$-homogeneous functions we have
\[
W_{\Bx}(\Bx,\By,\BF)\cdot\Bx+W_{\By}(\Bx,\By,\BF)\cdot\By=0
\]
identically.  Combining this straightforward  observation  with GCT, Theorem~\ref{th:GCT}, we obtain 
\begin{theorem}
  \label{th:clapeyron}
Suppose $\By(\Bx)$ is a Lipschitz stationary extremal of the functional (\ref{non-param}).
If the energy density function satisfies (\ref{scalefree}), then
  \begin{equation}
  \label{Clapeyron}
E[\By]=\nth{n}\int_{\dOm}\{\BP\Bn\cdot\By+\BP^{*}\Bn\cdot\Bx\}dS(\Bx),
\end{equation}
where $\BP\Bn\in L^{\infty}(\dOm;\bb{R}^{m})$ and $\BP^{*}\Bn\in
L^{\infty}(\dOm;\bb{R}^{n})$ are understood in the sense of traces.
\end{theorem}
 We refer to Theorem~\ref{th:clapeyron} as the
  Clapeyron-Eshelby Theorem (CET) since formula
  \eqref{Clapeyron} is a generalization  to
  arbitrary $m$ and $n$ and not necessarily smooth extremals of
  Eshelby's formula in 3D elasticity theory
  \cite[Eq.~(16)]{eshe70}. Interestingly, this crucial result
  was largely overlooked at the time, and was subsequently re-derived several times
    in different  forms, see for instance   \cite[formula~(2.10)]{green73},
    \cite[Eq.~(2.5)]{knst84}, \cite[Eq.~(38)]{hill86}). Moreover, it looks as if  in the subsequent derivations   the
    important role of configurational tractions, so
    prominent in the original Eshelby formulation,  was  obscured.
 
Observe that the first term on the right-hand side of
(\ref{Clapeyron}) can be rewritten  as the classical work of applied forces. Indeed, 
 we can rewrite (\ref{Clapeyron}) as
\begin{equation}
  \label{Clapeyron3}
E[\By]=\nth{n}\int_{\dOm}\{\BP\Bn\cdot\Bu+\BP_{\rm disp}^{*}\Bn\cdot\Bx\}dS(\Bx),
\end{equation}
where $\BP_{\rm disp}^{*}=W\BI_{n}-(\Grad\Bu)^{T}\BP$. Such rewriting   makes the first
term on the right-hand side of \eqref{Clapeyron3}  equal to  the
work of tractions $\BP\Bn$ over displacements $\Bu=\By-\Bx$ on the
boundary of the domain. However, the resulting Clapeyron-type relation  (\ref{Clapeyron}) remains 
fundamentally different from the classical CT
(\ref{Clap}), where the additional term on the \rhs\ of
(\ref{Clapeyron}) with the Eshelby's energy-momentum tensor $\BP^{*}$ is absent.

In this respect we  recall that  while Piola stress $\BP$  acts on the virtual
displacements $\Gd\By$ of the points on the boundary of the domain in
the actual (Eulerian) space, the Eshelby stress $\BP^{*}$  acts on the
virtual changes of the shape of the domain $\GO$ in the reference
(Lagrangian) space. The former is usually interpreted as the action of
``physical forces'', and the latter -- as the action of
``configurational forces'' (see for instance, \cite{gfa10}). Therefore,
the two terms in the surface integral in \eqref{increm} represent the
work of  physical and configurational forces. While the meaning of the
work of physical forces  does not require any clarifications, we show
in Section~\ref{sub:lbcf} that the work of configurational forces can
be associated with the  acquisition of a pre-stress during the
\emph{formation} of the body.

Another major difference between CT and CET is that in the latter the
multiple in front of the surface integral is $1/n$ instead of $1/2$, even in
3D space. As we have seen in Section~\ref{sub:CCT}, the
coefficient 1/2 emerges in the classical CT due to the 2-homogeneity
of the  quadratic energy density function of linear elasticity, while
no such homogeneity property is assumed in the more general nonlinear elastic
context of CET.

 \subsection{ CET follows  from a higher dimensional version of CT}
 \label{sub:GET_CT}
Next we show  that in some sense the $p$-homogeneous
Clapeyron-type Theorem~\ref{th:pGCT} can be viewed as implying the CET
Theorem~\ref{th:clapeyron}. This becomes clear if we realize that at
their source, the transformations (\ref{trgrp}) treat the
$(\Bx,\By)$-space $\mathfrak{X}$ as a single entity. 
Specifically, to see
the implied relation between CT and CET we need to go back to the description of the
problem in the extended space introduced in Section~\ref{sec:noether}.
We recall that in this space  the original Lagrangian $W(\Bx,\By,\BF)$ is  replaced by the
\emph{graph}-Lagrangian $\CW(\Bz,\BCF)$, given by (\ref{grL}).
Now, if $\CW$ does not depend explicitly on $\Bz$, i.e., $W=W(\BF)$, then
$\CW(\BCF)$ satisfies (\ref{phom}), with $p=n$. Therefore,
Theorem~\ref{th:pGCT} gives
\[
E[\By]=\CE[\Bz,D]=\nth{n}\int_{\Md D}\BCP\Bn\cdot\Bz(\Bt)dS(\Bt).
\]
Using formula (\ref{Phat12}) for $\BCP$ we obtain
\begin{multline*}
E[\By]=\nth{n}\int_{\Md D}\{
\BP^{*}(\Grad_{\Bx}\By(\BGn(\Bt)))\cof(\Grad_{\Bt}\BGn(\Bt))\BGv(\Bt)\cdot\BGn(\Bt)+\\
\BP(\Grad_{\Bx}\By(\BGn(\Bt)))\cof(\Grad_{\Bt}\BGn(\Bt))\BGv(\Bt)\cdot\By(\BGn(\Bt))
\}dS(\Bt),
\end{multline*}
where $\BGv(\Bt)$ is the outward unit normal to $\Md D$. Changing variables in
the above surface integral, using the relation
\[
\cof(\Grad\BGn(\Bt))\BGv(\Bt)dS(\Bt)=\Bn(\Bx)dS(\Bx),
\]
we obtain the CET represented by formula (\ref{Clapeyron}). We
emphasize that the defining commonality of all versions of  Clapeyron's Theorem is that they express the elastic energy of a stable configuration in terms of boundary tractions, meaning both physical and configurational tractions.

\subsection{CET for incompressible materials}
\label{sub:inc}
We now show that the  relation  \eqref{Clapeyron} can be  generalized to  cover  the case of scale-free energy with scale-free pointwise constraints. A typical example is incompressible elasticity, where every deformation $\By(\Bx)$ must satisfy the incompressibility constraint
\begin{equation}
  \label{incstr}
\det\Grad\By(\Bx)=1
\end{equation}
for almost all $\Bx\in\GO$. Specifically, one can prove the following theorem:
\begin{theorem}
  \label{th:constr}
  Suppose that the Lipschitz deformation $\By(\Bx)$ minimizes the
  energy (\ref{non-param}) and satisfies the $k$ pointwise constraints
  \begin{equation}
    \label{constr}
      C^{i}(\Bx,\By(\Bx),\Grad\By(\Bx))=0,\quad i=1,\ldots,k,
  \end{equation}
 where both  functions $W(\Bx,\By,\BF)$ and $C^{i}(\Bx,\By,\BF)$ are
 scale-free in the sense of (\ref{scalefree}). Then,
  \begin{equation}
  \label{constrClap}
  n\int_{\GO}W(\Bx,\By,\Grad\By)d\Bx= \int_{\dOm}\{\BP_{L}\Bn\cdot\By+\BP^{*}_{L}\Bn\cdot\Bx\}dS,
\end{equation}
where the \rhs\ is understood in the sense of traces, and where
$\BP_{L}$ and $\BP^{*}_{L}$ denote the Piola and Eshelby tensors corresponding to the Lagrangian
\[
L(\Bx,\By,\BF)=W(\Bx,\By,\BF)-\BGL(\Bx)\cdot\BC(\Bx,\By,\BF).
\]
\end{theorem}
\begin{proof}
From the general theory of constrained variational problems we know
that the energy minimizers must be stationary extremals of the
augmented energy functional
\[
\Tld{E}[\Hat{\By}]=\int_{\GO}\{W(\Bx,\By,\Grad\By)-\BGL(\Bx)\cdot\BC(\Bx,\By,\Grad\By)\}d\Bx,
\]
where $\hat{\By}=[\By;\BGL]\in\bb{R}^{m+k}$. While both the energy
density $W$ and the constraints $\BC$ are scale-free in the sense of
(\ref{scalefree}), the augmented Lagrangian
\[
\Tld{L}(\Bx,\Hat{\By},\Hat{\BF})=W(\Bx,\By,\BF)-\BGL\cdot\BC(\Bx,\By,\BF),\quad\Hat{\BF}=[\BF;\BG]\in\bb{R}^{(m+k)\times n}
\]
is not, as a function of $\Bx,\Hat{\By},\Hat{\BF}$, because it is affine in $\BGL$. 
Therefore we can apply Theorem~\ref{th:genoether} to the functional
$\Tld{E}[\Hat{\By}]$ and the variation
\[
\BX=e^{\Ge}\Bx,\quad\Hat{\BY}=[e^{\Ge}\By;\BGL].
\]
We compute
\[
\Tld{E}_{\Ge}[\Hat{\BY}]=\int_{e^{\Ge}\GO}\Tld{L}(\BX,\Hat{\BY}_{\Ge}(\BX),\Grad\Hat{\BY}_{\Ge}(\BX))d\BX,
\]
which, after the change of  variables $\Tld{\Bx}=e^{\Ge}\Bx$, becomes
\[
\Tld{E}_{\Ge}=\Ge^{n\Ge}\int_{\GO}\{W(e^{\Ge}\Bx,e^{\Ge}\By(\Bx),\Grad\By(\Bx))-
  \BGL(\Bx)\cdot\BC(e^{\Ge}\Bx,e^{\Ge}\By(\Bx),\Grad\By(\Bx))\}d\Bx.
\]
and since both $W$ and $\BC$ satisfy (\ref{scalefree}), we have
\[
\Tld{E}_{\Ge}=\Ge^{n\Ge}\int_{\GO}\{W(\Bx,\By(\Bx),\Grad\By(\Bx))-
  \BGL(\Bx)\cdot\BC(\Bx,\By(\Bx),\Grad\By(\Bx))\}d\Bx=\Ge^{n\Ge}\int_{\GO}Wd\Bx,
\]
where we took into account that $\By(\Bx)$ satisfies the constraints
$\BC(\Bx,\By(\Bx),\Grad\By(\Bx))=0$, as an extremal for $\Tld{E}[\Hat{\By}]$.
Therefore,
\[
\Gd\Tld{E}=n\int_{\GO}W(\Bx,\By(\Bx),\Grad\By(\Bx))d\Bx,
\]
and Noether's formula (\ref{increm}) becomes (\ref{constrClap}),
since
$\Gd\Bx=\Bx$, $\Gd\Hat{\By}=[\By;0]$, while
\[
  \BP_{\Tld{L}}=[\BP_{L};0],\quad\BP_{\Tld{L}}^{*}=\BP^{*}_{L}.
\]
\end{proof}
Applying this theorem to the incompressible elasticity constraint
(\ref{incstr}) and energy density $W(\BF)$ that does not depend on
$\Bx$ and $\By$ explicitly, we obtain, noting that both the energy
density and the constraint are scale-free,
\begin{equation}
  \label{incgrn}
  \int_{\GO}W(\Grad\By(\Bx))d\Bx=\nth{n}\int_{\dOm}\{(\BP-p(\Bx)\cof\Grad\By(\Bx))\Bn\cdot\By(\Bx)
  +(\BP^{*}+p(\Bx)\BI_{n})\Bn\cdot\Bx\}dS,
\end{equation}
where, following convention,  we have denoted the Lagrange multiplier $\GL(\Bx)$ by $p(\Bx)$ to indicate its physical meaning as pressure, and took into account that
\[
\BP^{*}_{L}=\BP^{*}-p(\Bx)\left[(\det\BF-1)\BI_{n}-\BF^{T}\cof(\BF)\right]=\BP^{*}+p(\Bx)\BI_{n}.
\]
Note that when $\By(\Bx)$ is an orientation preserving diffeomorphism, we can change variables $\By=\By(\Bx)$ in the first term on the \rhs\ of   \eqref{incgrn} and obtain
\begin{equation}
  \label{mixedgreen}
  \int_{\GO}W(\Grad\By)d\Bx=\nth{n}\int_{\dOm^{*}}(\BGs-p(\By)\BI_{n})\BN\cdot\By\,dS^{*}(\By)
  +\nth{n}\int_{\dOm}(\BP^{*}+p(\Bx)\BI_{n})\Bn\cdot\Bx\,dS(\Bx),
\end{equation}
where $\BGs(\By)$ is the Cauchy stress tensor, $\BN(\By)$ is the outward unit normal to $\dOm^{*}=y(\dOm)$, and $p(\By)=p(\Bx)$, with $\By=\By(\Bx)$.

\subsection{Shifting of the origin}
\label{sub:CCT1}

Next we  mention a straightforward result that will prove  useful in what follows. 
\begin{remark}
   \label{rem:shift}
   When $W$ does not depend explicitly on $\Bx$ and $\By$, the
   stationary extremals satisfy
\[
\int_{\dOm}\BP\Bn\,dS(\Bx)=\Bzr,\qquad\int_{\dOm}\BP^{*}\Bn\,dS(\Bx)=\Bzr.
\]
These relations allow one to re-write \eqref{Clapeyron} as
\begin{equation}
  \label{purple}
  \int_{\GO}W(\Grad\By)d\Bx=
\nth{n}\int_{\Md\GO}\{(\BP^{*}\Bn,\Bx-\Ba)+(\BP\Bn,\By(\Bx)-\Bb)\}dS(\Bx),
\end{equation}
where $\Ba\in\bb{R}^{n}$ and $\Bb\in\bb{R}^{m}$ are arbitrary constant
vectors.
\end{remark}
Formula (\ref{purple}) expresses a physically obvious fact
that when the energy density function does not depend on  $\Bx$ and
$\By$ explicitly, the choice of the coordinate system origins in both the reference and
the deformed configurations can be made arbitrarily and independently.

\subsection{Various facets of CET  in linear elasticity}
\label{sub:GCTlin}
To illustrate the obtained general formulas  it is instructive to
present  the classical Clapeyron's Theorem (\ref{Clap}) in relation to
 CET (\ref{Clapeyron}). 
Consider first a  general  linear elastic  energy density (\ref{linel}),
where $\BGve=(\BF+\BF^{T})/2$.
When the  body is also homogeneous, meaning that the tensor of elastic
moduli $\SFC$ is constant, the energy density function (\ref{linel})
is both 2-homogeneous and scale-free in the sense of (\ref{phom}),
with $p=2$, and (\ref{scalefree}), respectively. In this case  one can
also show that all equilibrium solutions are smooth and therefore, by
the Noether relation (\ref{Noether0}), also stationary. Hence, in
addition to (\ref{semiC}) that has the form (\ref{Clap}) in the
linearly elastic context, formula 
(\ref{Clapeyron}) is also applicable. Combining (\ref{semiC}) and
(\ref{Clapeyron}) permits us to give an alternative expression for the
energy in terms of the configurational stresses
\begin{equation}
  \label{semiL}
  E[\By]=\nth{n-p}\int_{\dOm}\BP^{*}\Bn\cdot\Bx\,dS, 
\end{equation}
 if $n\not=p$, and
\begin{equation}
  \label{n=p}
  \int_{\dOm}\BP^{*}\Bn\cdot\Bx\,dS=0,
\end{equation}
  if $n=p$.
In relation to three-dimensional linear elasticity, where $p=2$, and
$n=3$ formula (\ref{semiL}) is very different from the classical
CT (\ref{Clap}), in contrast to which formula (\ref{semiL}) involves the skew-symmetric part $w(\Bu)$ of $\Grad\Bu$,
\[
w(\Bu)=\frac{1}{2} (\Grad\Bu-(\Grad\Bu)^{T}),
\]
representing the infinitesimal rotations. 
\begin{equation}
  \label{Clapeyron223}
  \hf\int_{\GO}\av{\SFC e(\Bu),e(\Bu)}d\Bx=\int_{\dOm}\left\{
\hf\av{\BGs,e(\Bu)}\Bn\cdot\Bx-\BGs\Bn\cdot (e(\Bu)+w(\Bu))\Bx\right\}dS(\Bx).
\end{equation}
When $n=2$ we don't have a new alternative formula for the
energy. Instead, formula (\ref{n=p}) gives an (apparently new)
integral relation satisfied on the boundary by the antisymmetric part
of the gradient of an equilibrium solution:
\begin{equation}
  \label{neweq2}
  \int_{\dOm}\BGs\Bn\cdot w(\Bu)\Bx\,dS=\int_{\dOm}\left\{\hf\av{\BGs,e(\Bu)}(\Bn\cdot\Bx)-
  \BGs\Bn\cdot e(\Bu)\Bx\right\}dS.
\end{equation}
In general, one can eliminate the elastic energy from formulas
(\ref{Clap}) and (\ref{Clapeyron}) and obtain the analogue of
(\ref{neweq2}) in any number of dimensions:
\begin{equation}
  \label{neweq}
  \int_{\dOm}\BGs\Bn\cdot w(\Bu)\Bx\,dS=\int_{\dOm}\left\{\hf\av{\BGs,e(\Bu)}(\Bn\cdot\Bx)-
  \BGs\Bn\cdot e(\Bu)\Bx-\frac{n-2}{2}\BGs\Bn\cdot\Bu\right\}dS.
\end{equation}
Since (\ref{neweq}) holds in any subdomain occupied by the elastic
body, we can rewrite (\ref{neweq}) in differential form
\begin{equation}
  \label{newdeq}
  \Div(\BGs w(\Bu)\Bx)=\Div\left(\hf\av{\BGs,e(\Bu)}\Bx-\BGs e(\Bu)\Bx-\frac{n-2}{2}\BGs\Bu\right),
\end{equation}
giving the differential equation for $w(\Bu)$ in terms of $e(\Bu)$. In
two space dimension equation (\ref{newdeq}) simplifies with the \rhs\
depending only on $e(\Bu)$:
\[
\Div(\BGs\Grad\Bu\Bx)=\hf\Div(\av{\BGs,e(\Bu)}\Bx).
\]
To the best of our knowledge, formulas
(\ref{Clapeyron223})--(\ref{newdeq}) are new. It is likely that they
did not appear in the literature on linear elasticity due to the
rather striking appearance in them  of the skew-symmetric part
$w(\Bu)$ of the gradient $\Grad\Bu$.

\subsection{Physical meaning of $\BP^{*}$ loading}
\label{sub:lbcf}
In our main  Theorem  \ref{th:clapeyron}  we encounter two important tensor fields $\BP$  and $\BP^{*}$.  As we have already mentioned, while the   physical meaning of Piola stress $\BP$  is well known, the interpretation of the Eshelby stress $\BP^{*}$ is more subtle. To elucidate the  fundamental difference between the work done by $\BP$  and $\BP^{*}$ we present in this section
two illustrative examples focusing on the loading effected exclusively by configurational forces.

Consider first a one-dimensional body, given in Lagrangian coordinates as an interval $[0,L]$.
The stored energy is  
\begin{equation}
  \label{1Dex}
E[u]=\int_{0}^{L}W(u')dx,
\end{equation}
where $u(x)$ is the displacement field. We assume that the energy density function $W(\Gve)$ is a smooth, even, strictly convex function, such that $W'(0)=0$, corresponding to the absence of the physical pre-stress. The configurational pre-stress is encoded by the residual energy $W_{*}=W(0)>0$. 
The relationship between the physical (Piola) stress
 \begin{equation}
P=W'(\Gve) 
\end{equation}
and the configurational (Eshelby) stress
\begin{equation}
\label{P*1D}
  P^{*}=W(\Gve)-\Gve W'(\Gve),
\end{equation}
 where the strain is  $\Gve=u'(x)$  is shown in the left panel of Fig.~\ref{fig:Clap1D}.
\begin{figure}[t]
  \centering
  \subfigure[Generic applied strain $\Gve$.]{
    \includegraphics[scale=0.25]{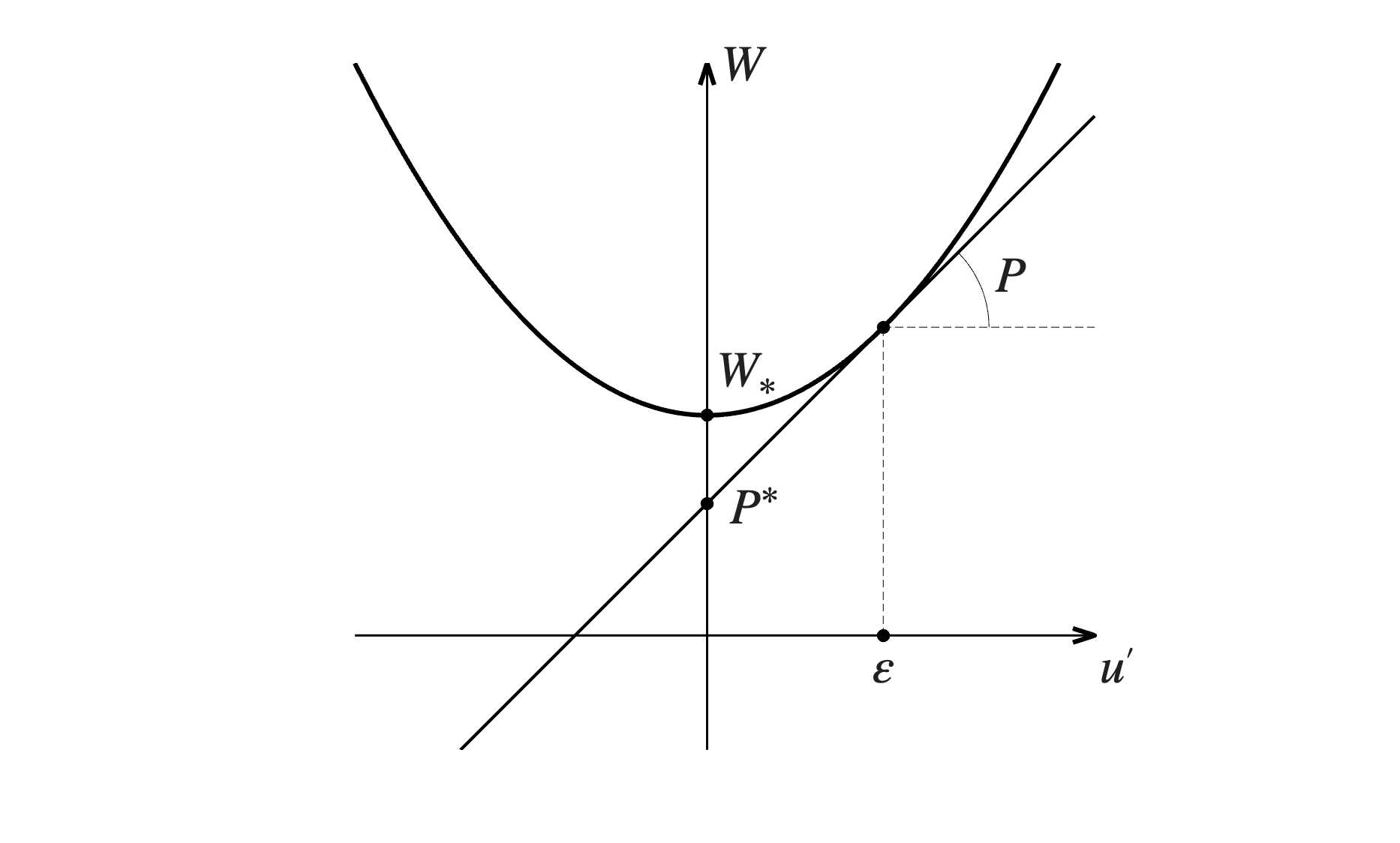}}
    \subfigure[``Optimal'' applied strain $\Gve_{\rm opt}$,
    corresponding to $P^{*}=0$.]{
      \includegraphics[scale=0.25]{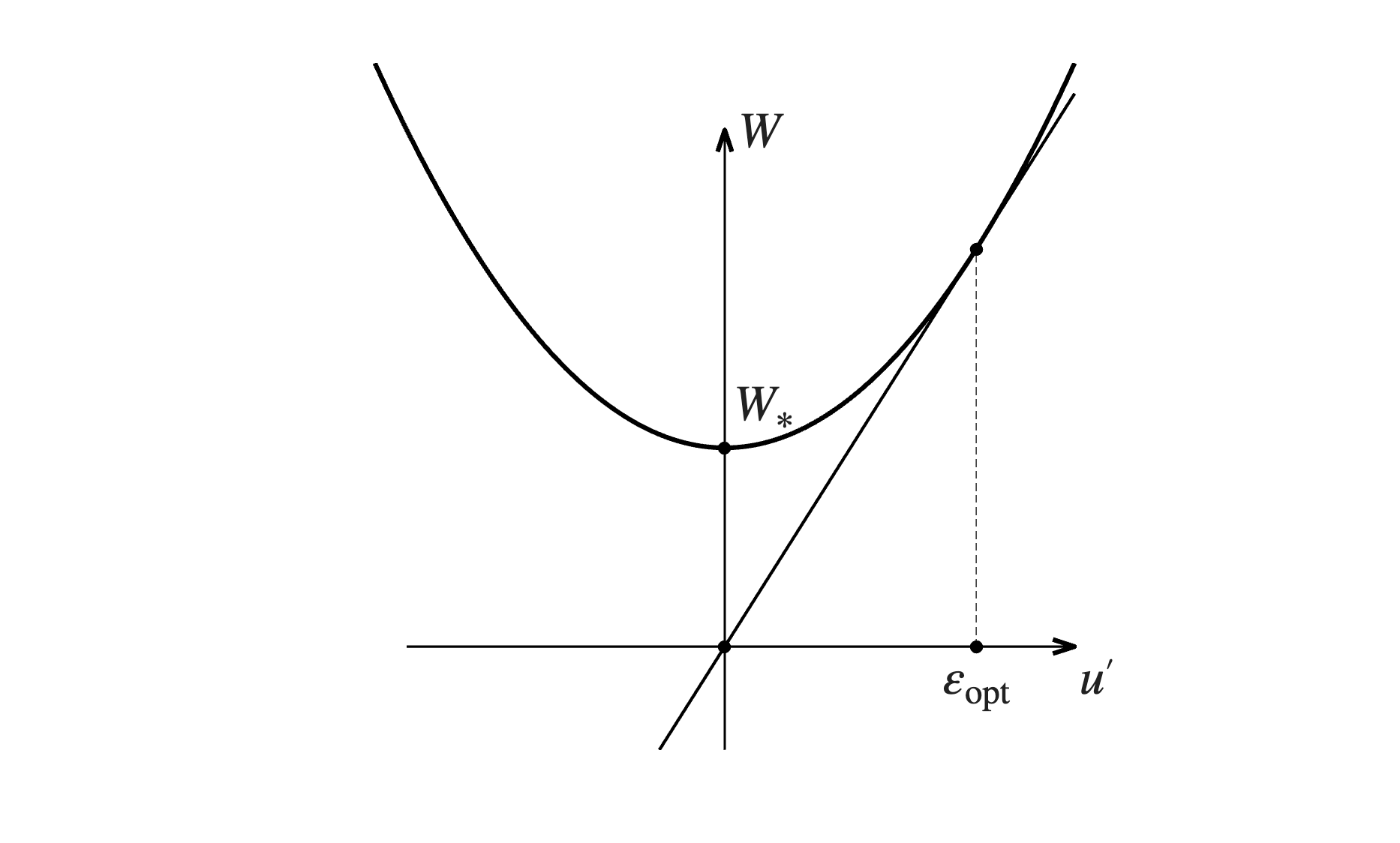}}
  \caption{Relations between physical stress $P$, strain $\Gve$ and the configurational stress $P^{*}$ in equilibrium.}
  \label{fig:Clap1D}
\end{figure}
Behind  this figure is the observation that in equilibrium we always have $(W'(u'))'=0$ and due to the strict monotonicity of $W'(\Gve)$, expressing the strict convexity of $W(\Gve)$, all equilibrium configurations are the ones where $\Gve=u'$ is constant on $[0,L]$. 

Consider now the case when  the deformation at points on $\dOm$ is prescribed,
but the domain $\GO$ itself is not fixed. Specifically , we assume that 
$U_{0}=u(0)$ and $U_{1}=u(L)$  with  $U_{0}$, $U_{1}$ given, while  allowing the reference length  $L$ to  vary.  In other words,  our pre-stressed body  is \emph{allowed to grow}.
Under these conditions the admissible outer variations $\Gd\By$ in (\ref{deriv}) must vanish on $\dOm$, but the boundary values of admissible inner variations $\Gd\Bx$ are no longer restricted in any way. Then the vanishing of the first variation $\Gd E$ of the energy is equivalent to the Euler-Lagrange equation (\ref{ELeq}) and the Noether  equation (\ref{stationary}) augmented with the displacement boundary conditions together with the condition 
\begin{equation}
  \label{optimshape}
  \BP^{*}\Bn=0\text{ on }\dOm.
\end{equation}
This  leads to a free boundary problem, where the size/shape of
$\GO$ must be determined in parallel with the deformation $\By(\Bx)$
to minimize the energy.
\begin{figure}[t]
  \centering
  \includegraphics[scale=0.25]{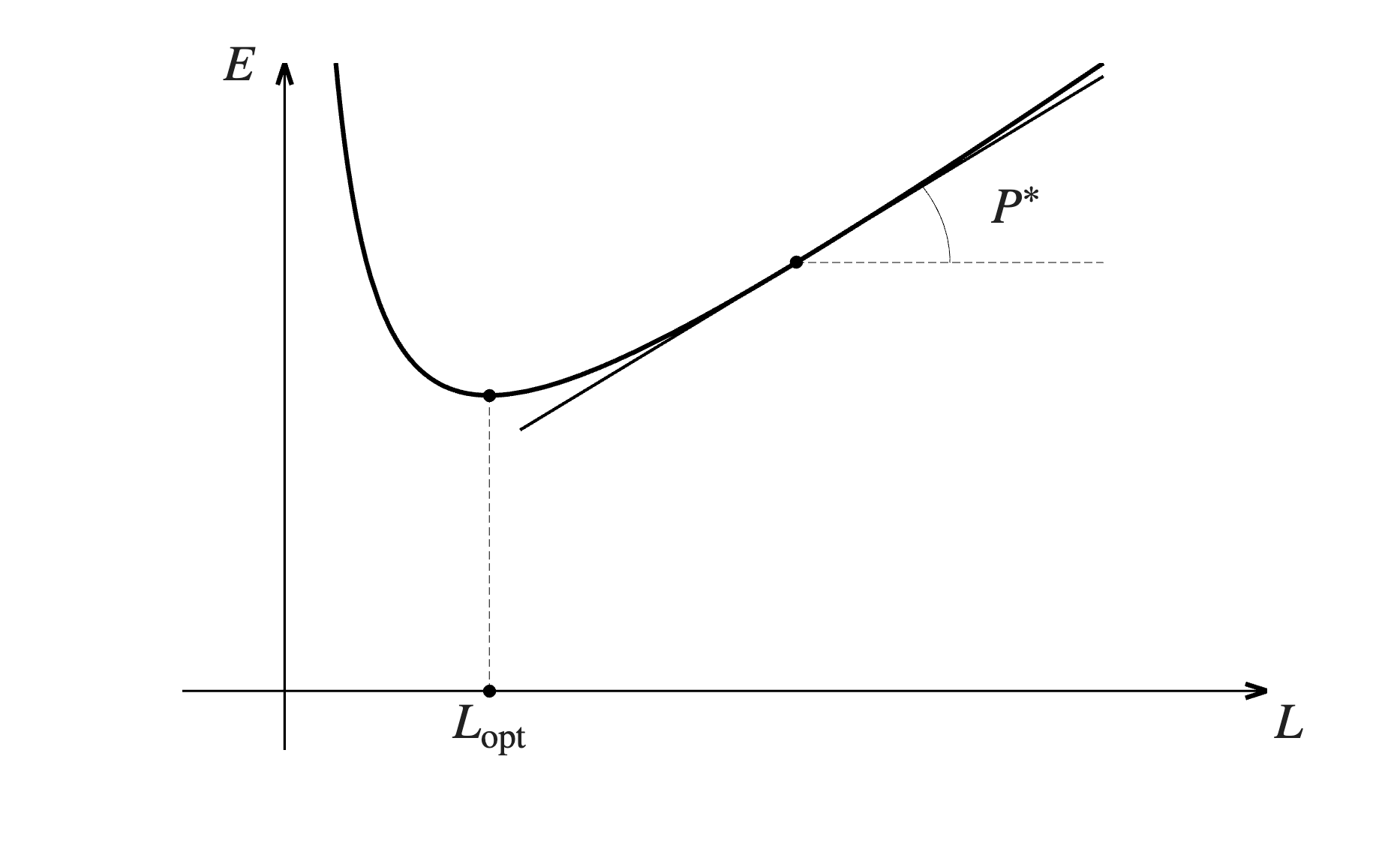}
  \caption{Energy as a function of the domain shape, when boundary displacements are fixed.}
  \label{fig:EofL}
\end{figure}
To clarify the meaning of the boundary condition \eqref{optimshape} we
now show that in our example (\ref{1Dex}) the energy extremum as a
function of the size of the domain $L$ is indeed attained when $P^{*}=0$. 
Indeed, if $U_{0}$ and $U_{1}$ are prescribed, then, by strict convexity of $W(\Gve)$, there is a unique value of the constant equilibrium strain
\[
\Gve(x)=\frac{U_{1}-U_{0}}{L}.
\]
Then, the energy as a function of $L$ is
\[
E=LW\left(\frac{U_{1}-U_{0}}{L}\right).
\]
It is easy to see that
\begin{equation}
  \label{optimL}
  \frac{dE}{dL}=P^{*},\qquad\frac{d^{2}E}{dL^{2}}=\frac{\Gve^{2}W''(\Gve)}{L}\geq 0, 
\end{equation}
which shows that $P^{*}=0$ is indeed the condition of minimality
of $E(L)$, assuming that the $C^{2}$ energy density $W(\Gve)$ is convex on $\bb{R}$
and grows super-linearly at infinity.
Then, if $\Gve_{\rm opt}$ is the value of the  optimal strain shown in Fig.~\ref{fig:Clap1D}(b)),
the corresponding optimal length
\[
L_{\rm opt}=\frac{U_{1}-U_{0}}{\Gve_{\rm opt}}
\]
delivers the minimum to the energy.
The fact that at $P^{*}\neq 0$ the system can still lower its energy
by either growth or resorption is clear from Fig.~\ref{fig:EofL}. 

Our second example illustrates the fact that an elastic body may
have a nonzero energy even if it is not loaded externally by
``physical" surface  tractions in the sense that the \bc
\[
  \BP\Bn=0,\quad\Bx\in \partial \Omega
\]
is imposed. According to GCT, this would  imply
that the body is still loaded by the ``configurational" surface
tractions $\BP^{*}\Bn$. The idea is to construct 
a nontrivial stationary extremal of the functional  
\[
  E[\Bu]=\int_{\Omega} W(\Bx,e(\Bu))d\Bx,
\]
defined on the domain $\Omega\subset\bb{R}^{n}$; the notation for linear elastic strain $e(\Bu)$ was defined in (\ref{eofu}). To this end it will sufficient to consider the simplest linearly elastic quadratic energy density of the form  
\begin{equation}
  \label{Wdisloc}
  W(\Bx,\BGve)=\hf|\BGve -\BGve_0(\Bx)|^{2},
\end{equation}
where $\BGve$ is a symmetric $n\times n$ matrix, and the inhomogeneity $\BGve_0(\Bx)$ is incompatible in the sense that 
\begin{equation}
  \label{Wdisloc1}
{\rm  Ink} [\BGve_0(\Bx)] \neq 0,
\end{equation}
where
\[
  {\rm  Ink} [\BGve(\Bx)]_{ijkl}=\mix{\Gve_{ij}}{x^{k}}{x^{l}}-\mix{\Gve_{ik}}{x^{j}}{x^{l}}+
  \mix{\Gve_{lk}}{x^{i}}{x^{j}}-\mix{\Gve_{lj}}{x^{i}}{x^{k}}.
\]
It is  clear  that  \eqref{Wdisloc1} eliminates the trivial solution
$\BGve(\Bx)=\BGve_{0}$ with zero energy. In fact, as we show below,
the configurational prestress can be attributed exclusively to the
incompatibility of the embedded inhomogeneity $\BGve_{0}(\Bx)$.

We begin with recalling the  orthogonal Beltrami-Donati decomposition
(e.g. \cite{gekr09}) of the Hilbert space $\CH=L^{2}(\GO;\Sym(\bb{R}^{n}))$
\begin{equation}
  \label{decom}
  \CH=\CE\oplus\CJ,\quad \CE=\{e(\Bu): \Bu\in
H^{1}(\GO;\bb{R}^{n})\},\ \CJ=\{\BGs\in\CH:\Div\BGs=0,\ \BGs|_{\dOm}\Bn=0\}.
\end{equation}
In view of  \eqref{decom} we can  write any field in $\CH$ as a sum of a field in $\CE$ and a field in $\CJ$. Specifically, for a generic prestrain $\BGve_{0}\in\CH$ we can find $\Bu\in H^{1}(\GO;\bb{R}^{n})$ and $\BGs\in\CJ$, such that
\begin{equation}
  \label{decomp}
  \BGve_{0}(\Bx)=e(\Bu)-\BGs. 
\end{equation}
Then, it is straightforward to show that  the energy minimizer in the elastic problem involving  physically unloaded body is $\Bu(\Bx)$, the associated residual stress is  $\BGs=e(\Bu)-\BGve_{0}(\Bx)$, and the nontrivial minimum value of the elastic energy is 
\[
E[\Bu]=\hf\|\BGs\|^{2}_{L^{2}(\GO)}.
\]
The incompatible component $\BGs(\Bx)$ of pre-strain $\BGve_{0}(\Bx)$ can be found by solving the system
\begin{equation}
  \label{Inkeq}
  \begin{cases}
    {\rm Ink}[\BGs]=-{\rm Ink}[\BGve_{0}],&\Bx\in\GO\\
    \Div\BGs=0,&\Bx\in\GO,\\
    \BGs\Bn=0,&\Bx\in\dOm.
  \end{cases}
\end{equation}
 The decomposition (\ref{decom}) implies that the \bvp\ (\ref{Inkeq}) has a unique solution on simply connected
domains, since Ink$[\BGve]=0$ implies $\BGve=e(\Bv)$ for some $\Bv\in
H^{1}$. Thus, $\BGs$, and therefore the energy $E[\Bu]$ depends only on ${\rm Ink}[\BGve_{0}]$. In fact, it is straightforward to show that 
\[
\BP^{*}\Bn=W\Bn=\hf|\BGs|^{2}\Bn,\text{ on }\dOm,
\]
which suggests  that our ``physically unloaded'' body  was in fact  ``configurationally loaded.'' 

To illustrate the solutions we consider the problem with spherical symmetry. Specifically,
the goal will be  to construct a nontrivial radial stationary extremal of the functional 
\begin{equation}
  \label{Erad}
  E[\Bu]=\int_{B(0,R)}W(\Bx,e(\Bu))d\Bx,
\end{equation}
defined on the ball $B(0,R)\subset\bb{R}^{n}$ and with the energy density $W$ defined by \eqref{Wdisloc}. 
We assume further that 
\begin{equation}
 \label{Erad1}
  \BP\Bn=0,\quad|\Bx|=R. 
\end{equation}
 To ensure both  spherical symmetry and scale invariance  we choose
 the eigenstrain $\BGve_{0}(\Bx)$ in the form 
\begin{equation}
  \label{tracfree3}
  \BGve_0(\Bx)=a\tns{\Hat{\Bx}},\quad\Hat{\Bx}=\frac{\Bx}{|\Bx|}.
  \end{equation}    
  It is then straightforward to check that 
 \begin{equation}
  \label{tracfree31}  
  {\rm Ink}[\BGve_0]_{ijkl}(\Bx)=\frac{2a}{r^{2}}(\hat{x}^{l}\hat{x}^{k}\Gd_{ij}-
  \hat{x}^{i}\hat{x}^{k}\Gd_{lj}-\hat{x}^{l}\hat{x}^{j}\Gd_{ik}+\hat{x}^{i}\hat{x}^{j}\Gd_{lk}
  +\Gd_{ik}\Gd_{jl}-\Gd_{ij}\Gd_{kl}),
\end{equation}
except when $n=2$, in which case ${\rm Ink}[\BGve_0](\Bx)$ is a
multiple of a Dirac delta-function; the latter condition  means
that there is a local displacement field in
$\bb{R}^{2}\setminus\{0\}$, corresponding to the strain field
\eqref{tracfree3}, but no global one. We can regard
${\rm Ink}[\BGve_0]$ as a linear operator on the space of $n\times n$
matrices acting by the rule
\[
({\rm Ink}[\BGve_0]\BGx)_{ij}={\rm Ink}[\BGve_0]_{ijkl}\xi^{kl},
\]
where the summation over repeated indices is implied. In that case,
formula (\ref{tracfree31}) can also be written as
\[
{\rm Ink}[\BGve_0]\BGx=\frac{2a}{r^{2}}\left(\Tld{\BGx}-\Trc(\Tld{\BGx})\Tld{\BI}\right),\quad
\Tld{\BGx}=\Tld{\BI}\BGx\Tld{\BI},\quad
\Tld{\BI}=\BI-\tns{\Hat{\Bx}},\quad n>2,
\]
and 
\[
  {\rm Ink}[\BGve_0]\BGx=2\pi
  a\Gd(\Bx)\left(\BGx-\Trc(\BGx)\BI\right),\quad n=2.
\]
As the   trivial candidate $e(\Bu)=\BGve_{0}$  is not admissible,   the minimizer of the energy  \eqref{Erad} with free \bc s  \eqref{Erad1}  will be a stationary extremal possessing positive energy as it will be unstressed ``physically'' but still prestressed ``configurationally'' in the sense that it will carry nonzero residual stresses.
 The minimizer solves the  traction boundary value problem \eqref{Inkeq} which in our special case takes the form
  \[
    \begin{cases}
    \Grad(\Div\Bu)+\GD\Bu=2\Div\BGve_{0}(\Bx),&\Bx\in\GO,\\
    \BGs\Bn=(e(\Bu)-\BGve_{0})\Bn=0,&\Bx\in\dOm,
    \end{cases}
  \]
  obtained by taking the divergence of (\ref{decomp}) and restricting
  the decomposition (\ref{decomp}) to $\dOm$.
To obtain such a minimizer explicitly,  when $\GO=B(0,R)$---the ball in $\bb{R}^{n}$, centered at 0 with radius R, and $\BGve_{0}(\Bx)$ given by (\ref{tracfree3}) we first  rescale both $\Bx$ and $\Bu$ variables and set, \WLOG, $|a|=1$, $R=1$.
We will look for a radial solution
\begin{equation}
  \label{urad}
  \Bu(\Bx)=\eta(r)\Hat{\Bx}. 
\end{equation}
Since $\Grad\Bu$ is symmetric, the 
 function $\eta(r)$  should  be chosen such that
\begin{equation} \label{r01}
\GD\Bu=a\Div(\tns{\Bx}),
\end{equation}
or in terms of $\eta(r)$, such that
\begin{equation} \label{r0}
\eta''(r)+(n-1)\left(\frac{\eta(r)}{r}\right)'=\frac{a(n-1)}{r}.
\end{equation}
A solution of \eqref{r0} that does not blow up at $r=0$ must have the form   
\begin{equation}
  \label{eta}
    \eta(r)=a\left(br+\frac{n-1}{n}r\ln r \right).
\end{equation}
Finally, the zero traction boundary condition  \eqref{Erad1} gives
$b=1/n$, yielding
\begin{equation}
  \label{u}
  \Bu(\Bx)=\frac{a\Bx}{n}(1+(n-1)\ln|\Bx|)
\end{equation}
With an exact solution at hand  we can also  compute the energy of our configurationally prestressed body by direct substitution, obtaining
\begin{equation}
  \label{Epos1}
E[\Bu]= \frac{a^{2}(n-1)}{2n^{2}}|B(0,1)|>0.
\end{equation}
Here $|B(0,1)|$ denotes the $n$-dimensional volume of the unit ball in $\bb{R}^{n}$. Moreover we can   compute explicitly the expressions for the tensor fields $\BP(\Bx)$ and  $\BP^{*}(\Bx)$ obtaining 
\begin{equation}
  \label{P}
  \BP(\Bx)=\BGs(\Bx)=a\left(1+\frac{n-1}{n}\ln r\right)\BI_{n}-a\tns{\Hat{\Bx}}, 
\end{equation}
\begin{multline}
  \label{P*}
    \BP^{*}=a\left(\frac{(n+2)(n-1)^{2}\ln r -1}{n^{2}}\tns{\Hat{\Bx}}+\right.\\
  \left.\frac{(n+2)(n-1)^{2}(\ln r)^{2}+2(n^{2}-1)\ln r +n+1)}{2n^{2}}\BI_{n}\right). 
\end{multline}
The plots of the radial and circumferential components of
$\BP$ and $\BP^{*}$, given by (\ref{P}) and (\ref{P*}), respectively,
are shown in Fig.~\ref{fig:radial} for $n=3$. Panels (a), (b)
correspond to $a=1$ and panels (c), (d) correspond to $a=-1$. 
Note  the nonzero value of the radial component of $\BP^{*}$ on the boundary. 
\begin{figure}[h!]
  \centering
  \subfigure[]{
    \includegraphics[scale=0.2]{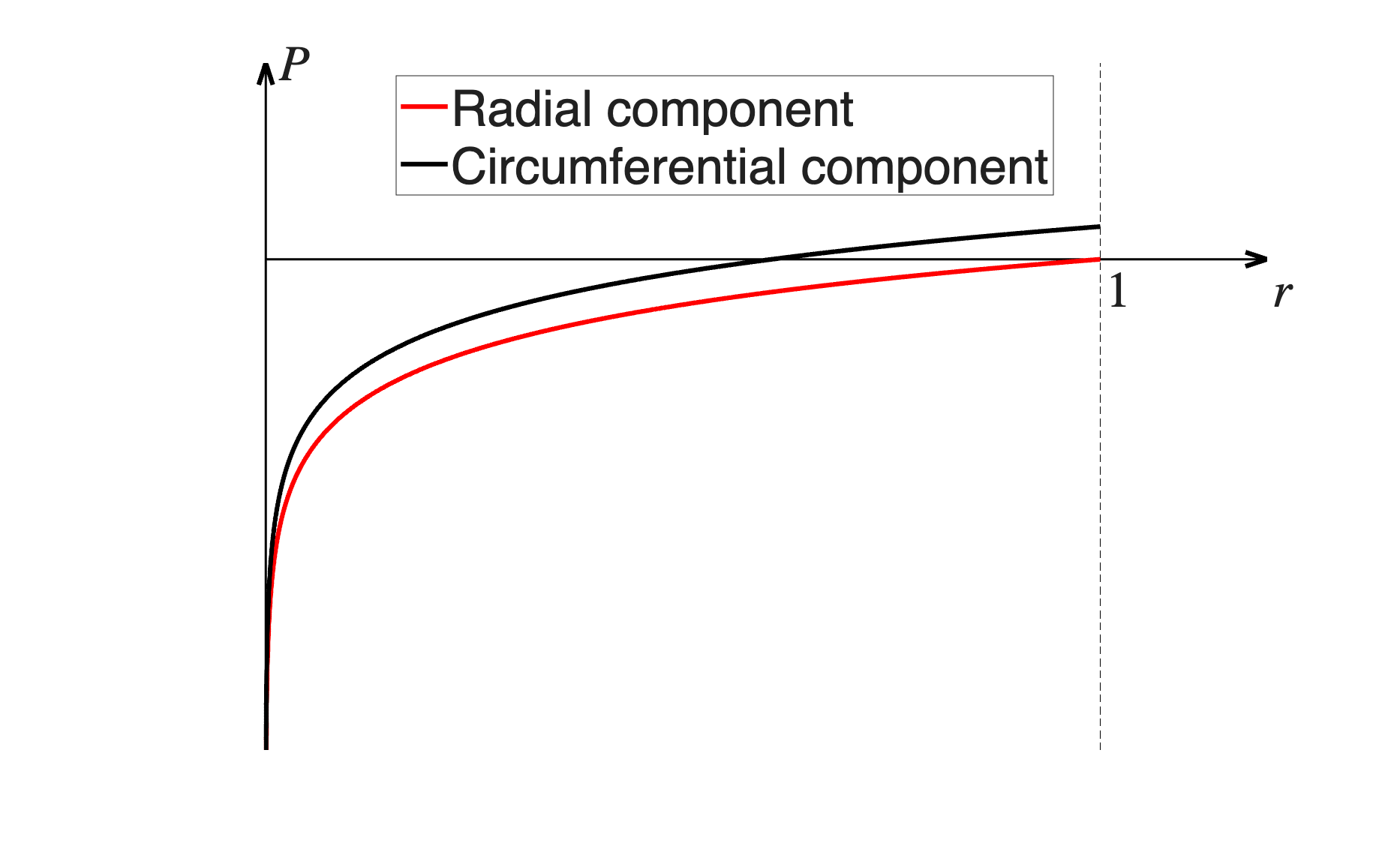}}
  \subfigure[]{
    \includegraphics[scale=0.2]{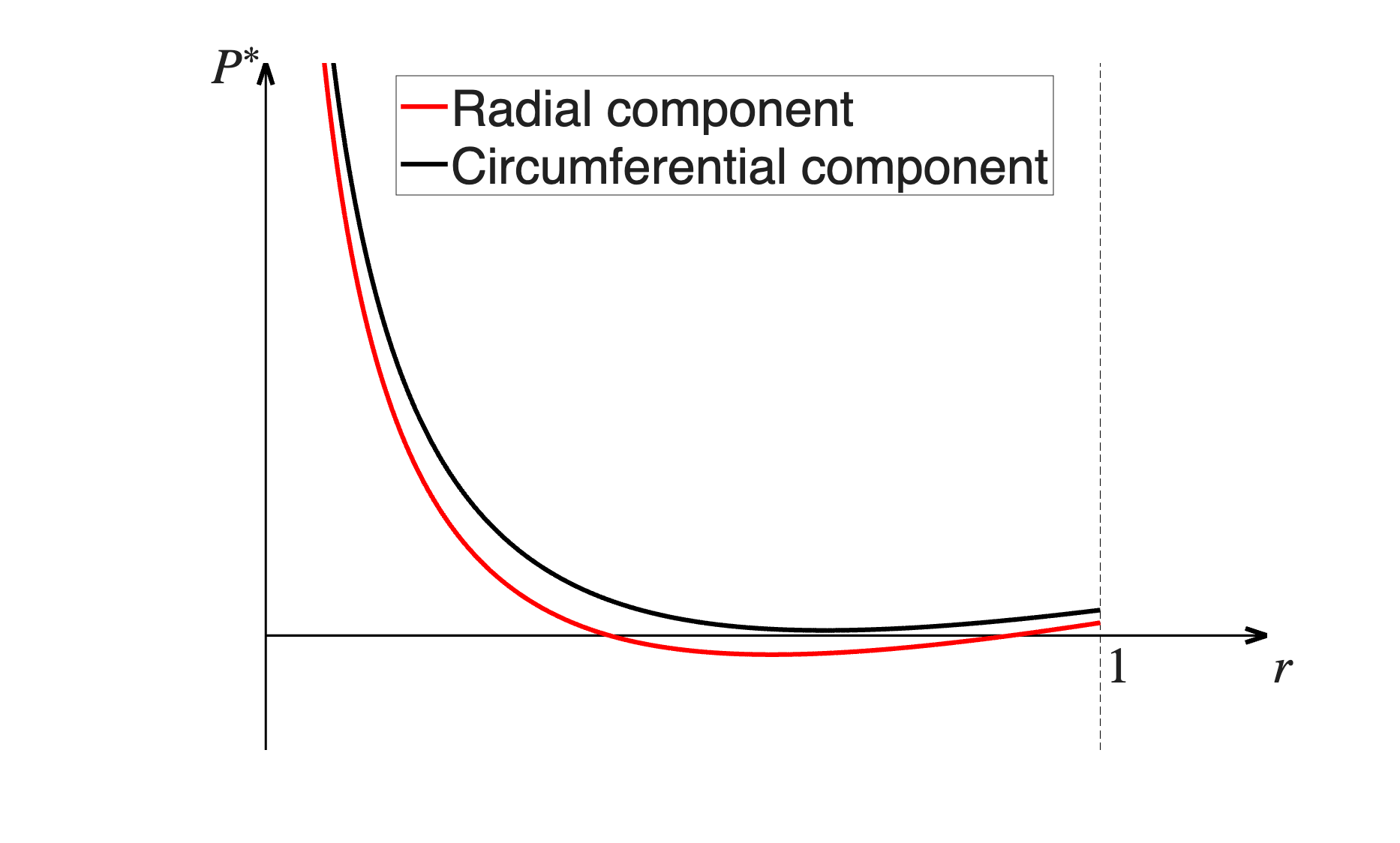}}\\
  \subfigure[]{
    \includegraphics[scale=0.2]{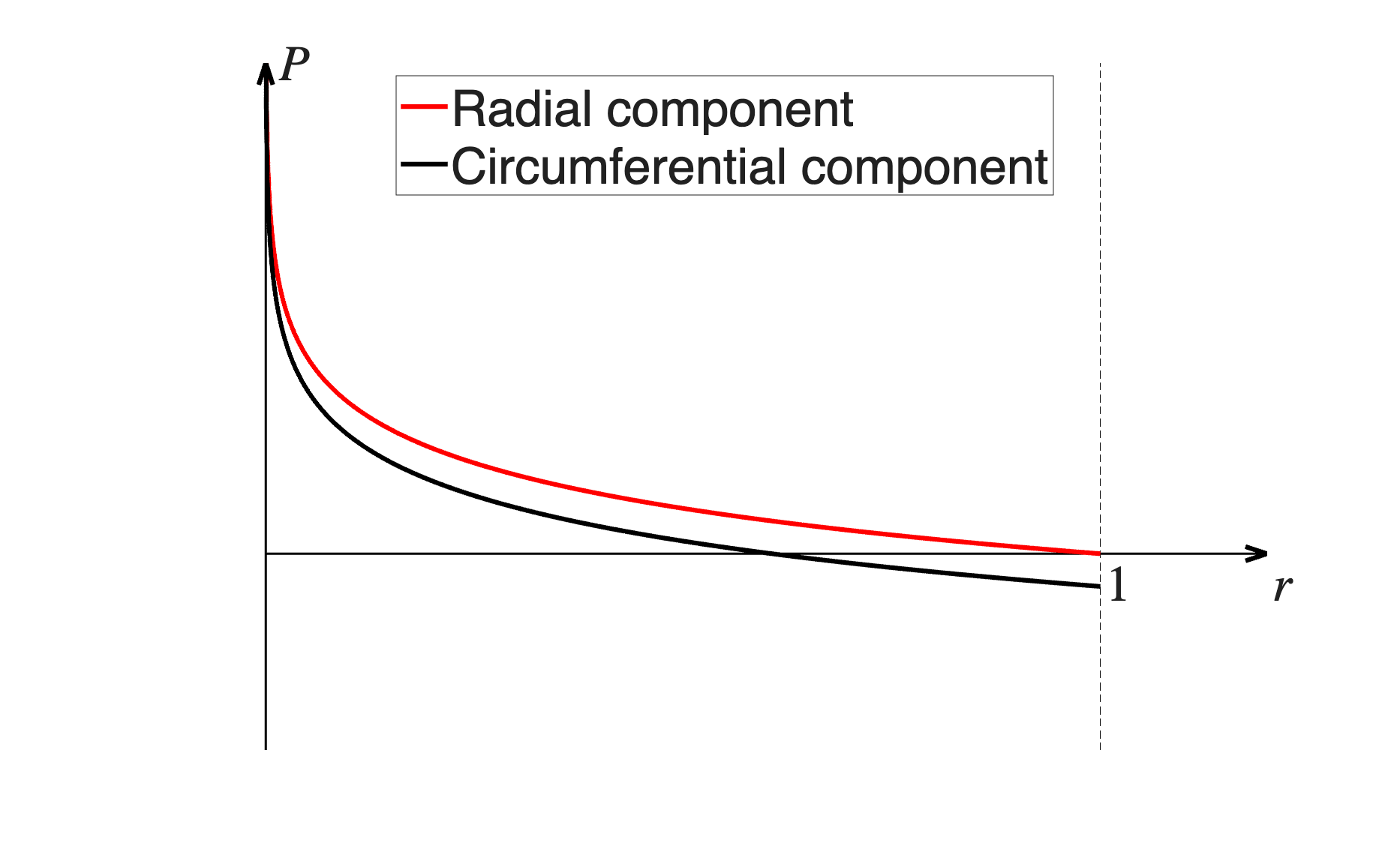}}
  \subfigure[]{
    \includegraphics[scale=0.2]{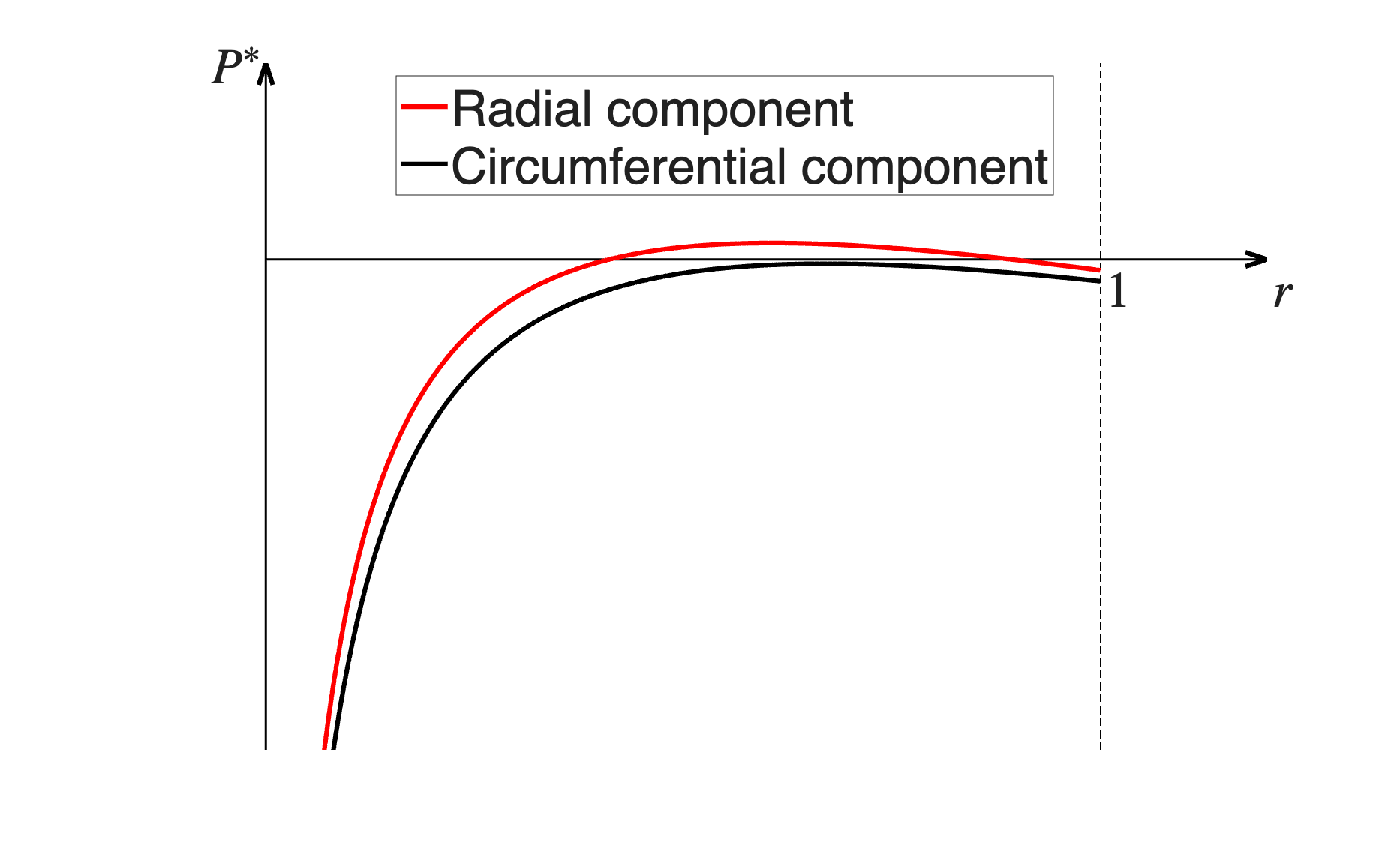}}
  \caption{Radial and circumferential components of $\BP$, given by (\ref{P}) (panels (a), (c)),
    and $\BP^{*}$, given by (\ref{P*}) (panels (b), (d)), corresponding to $a=1$ (panels
    (a), (b))  and $a=-1$ (panels (c), (d)).}
  \label{fig:radial}
\end{figure}

 \begin{figure}[h!]
 \centering
 \subfigure[]{
   \includegraphics[scale=0.1]{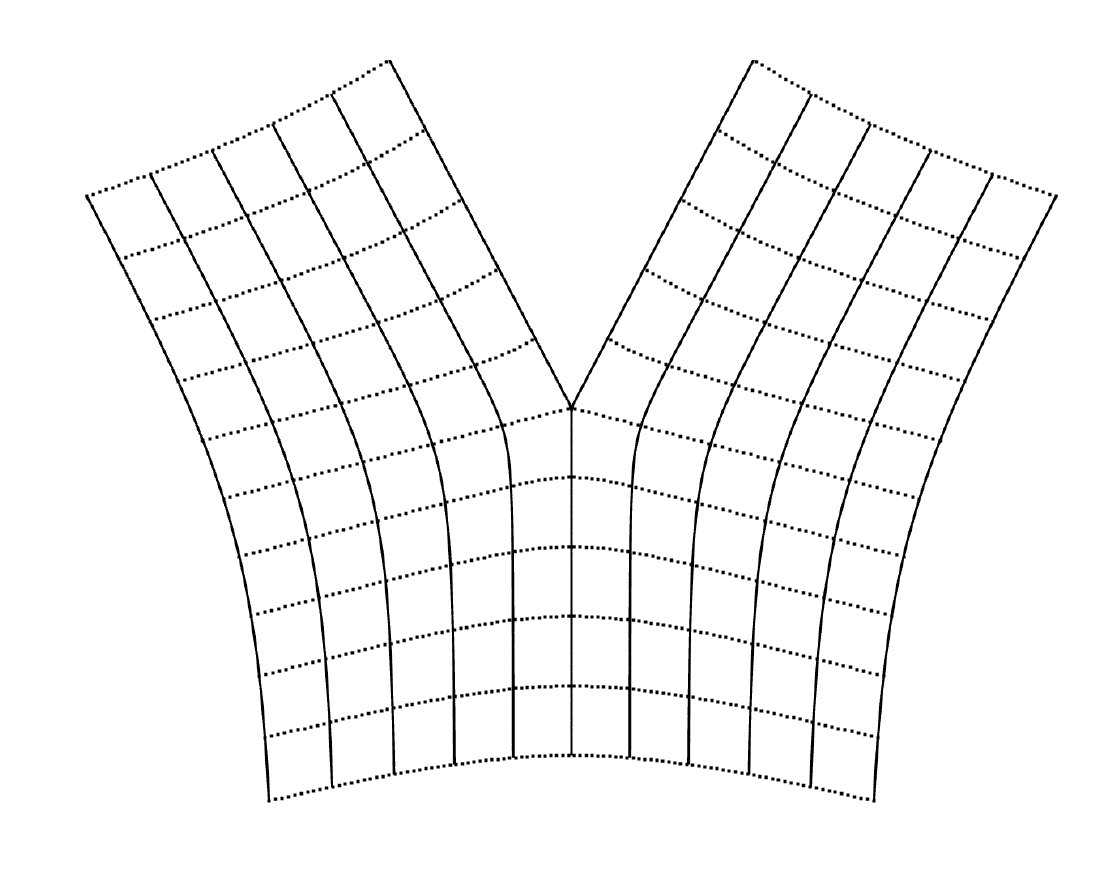}}\hspace{5ex}
 \subfigure[]{
   \includegraphics[scale=0.1]{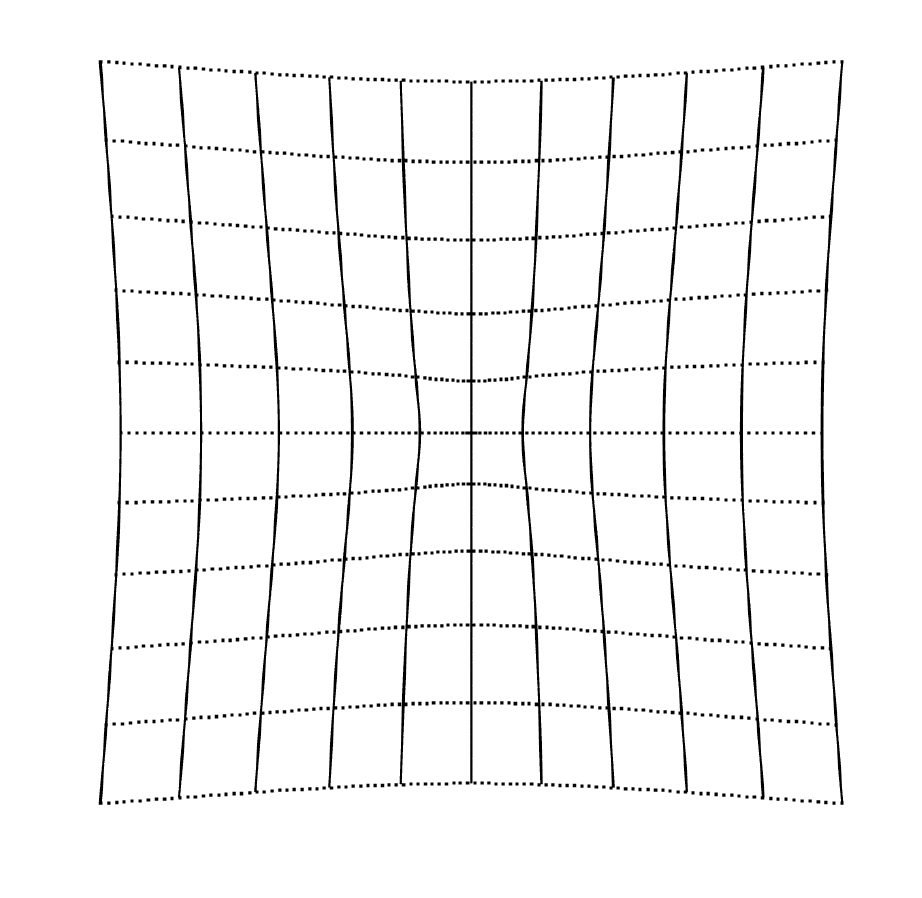}}\\
 \subfigure[]{
   \includegraphics[scale=0.08]{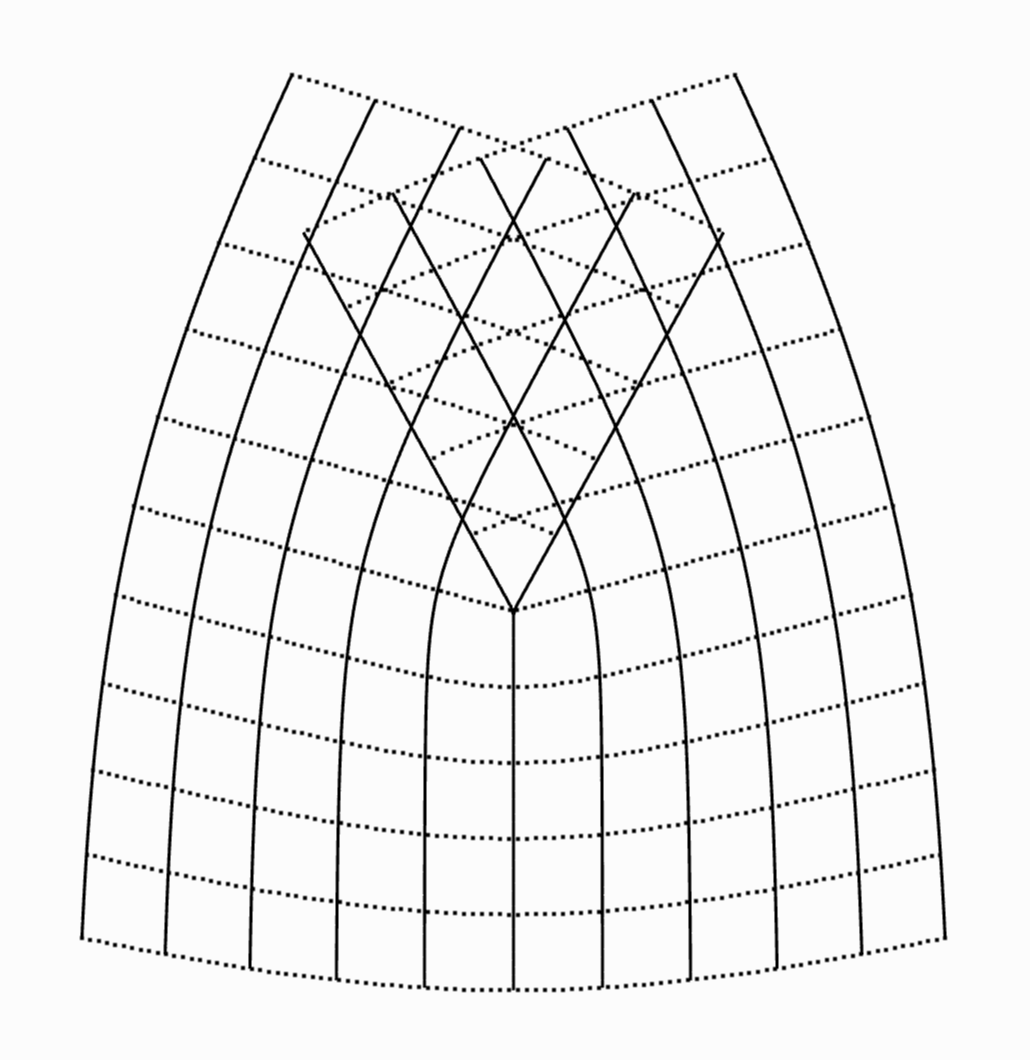}}\hspace{9ex}
 \subfigure[]{
   \includegraphics[scale=0.08]{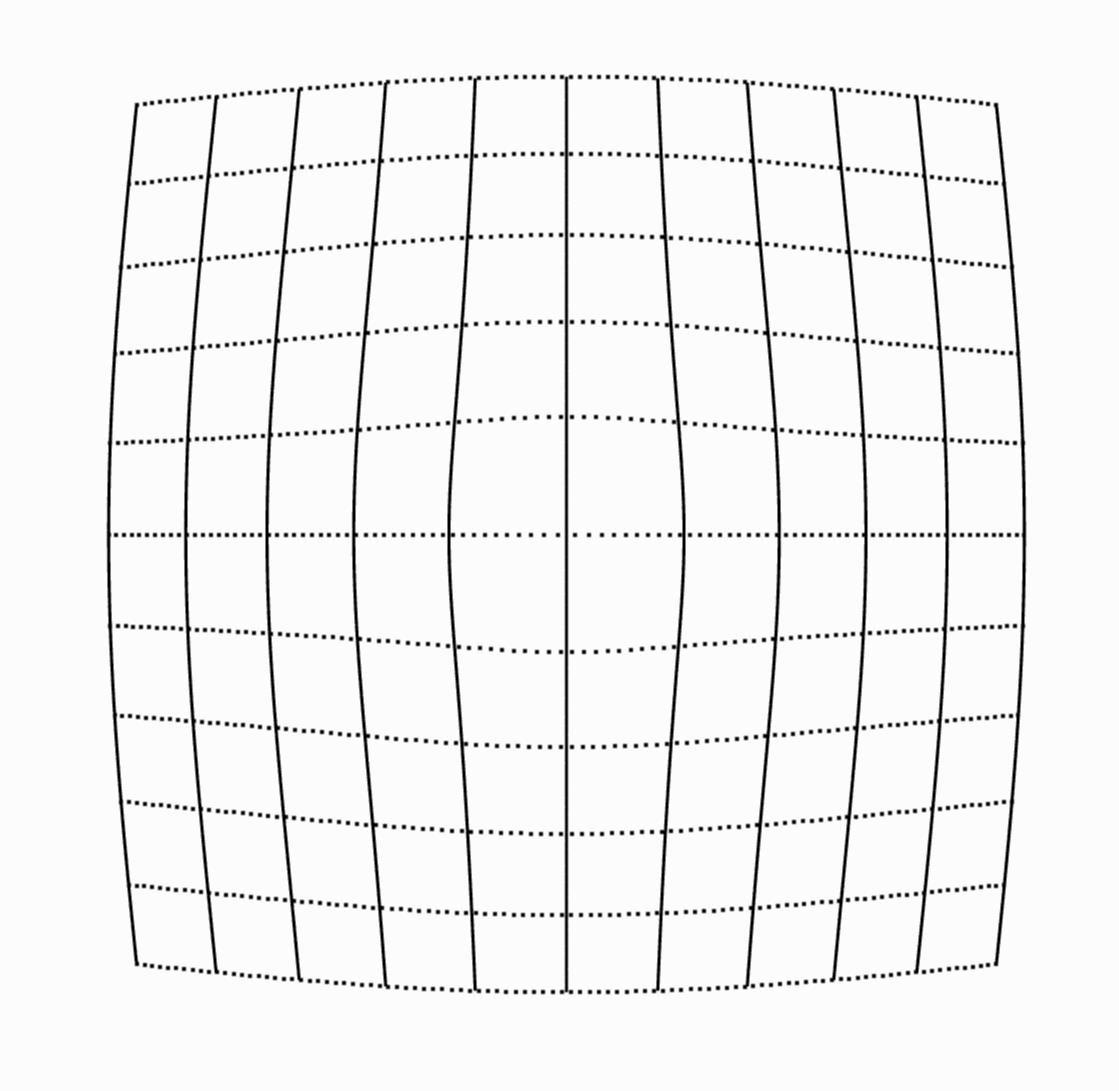}}
 \caption{(a), (c): Images of a perfect 2D grid corresponding to the
 discontinuous displacement field $\Bu(\Bx)$, given by (\ref{P2}),
 with $a=1$ and $a=-1$, respectively. (b), (d): Images of a perfect 2D grid corresponding to the energy-minimizing displacement field $\Bu(\Bx)$, given by (\ref{u}) with $a=1$ and $a=-1$, respectively.
}
 \label{fig:discl}
\end{figure}

More detailed results are illustrated in Fig.~\ref{fig:discl} in  the
case $n=2$. There we  show  the images of the Cartesian grid lines
under the deformations $\By(\Bx)=\Bx+\Ge\Bu(\Bx)$, for
$\Ge=0.2$ and $a=1$ (panels (a), (b)) and $a=-1$ (panels (c), (d)),
respectively. Panels (b), (d) correspond to the energy minimizing displacement
(\ref{u}), while panels (a), (c) correspond to the discontinuous
displacement field
\begin{equation}
  \label{P2}
  u_{1}(x_{1},x_{2})=a\left(x_{1}+x_{2}\arctan\left(\frac{x_{2}}{x_{1}}\right)\right),\quad
  u_{2}(x_{1},x_{2})=a\left(x_{2}-x_{1}\arctan\left(\frac{x_{2}}{x_{1}}\right) \right),
\end{equation}
generated by the kinematically incompatible
strain field (\ref{tracfree3}). In polar coordinates the displacement
field (\ref{P2}) can be written in a more compact  form $\Bu=a(r\Be_{r}-r\Gth\Be_{\Gth})$.

As we have already mentioned, the positive value of the energy should   be attributed to the
nonzero value of the external configurational loading
\[
\BP^{*}(\Bx)\Bn(\Bx) =\frac{a(n-1)}{2n^{2}}\Hat{\Bx},\quad|\Bx|=R=1.
\]
Such loading  can be thought of as the  factor ensuring the   embedding of the incompatibility Ink$[\BGve_{0}]$ into the body during the process of its imaginary growth through surface deposition. Given that we  know  the solution of the  energy minimization problem explicitly, we can   independently 
 compute the  energy by applying  
Theorem~\ref{th:clapeyron}  
 \begin{equation}
  \label{Clapeyron1}
  E[\Bu]=\nth{n}\int_{\dOm}\BP^{*}\Bn\cdot\Bx\,dS=
  \nth{n}\int_{\dOm}W(\Bx,\By,\BF)\Bn\cdot\Bx\,dS.
\end{equation}
Here we used the fact that   in our example ``physical''  tractions were absent   and therefore 
\[
  \BP^{*}\Bn=W(\Bx,\By,\BF)\Bn.
\]
Note also that Theorem ~\ref{th:clapeyron}  is applicable 
 because the energy density (\ref{Wdisloc}) with
  (\ref{tracfree3}) is scale-free, i.e., satisfies (\ref{scalefree}). 
  In our special case
formula (\ref{Clapeyron1}) gives
\[
  E[\Bu]=\frac{R}{n}\int_{\Md B(0,R)}W(\Bx,\Grad\By)dS=
  \frac{R|\Md B(0,R)|}{2n}\left((\eta'(R)-a)^{2}+(n-1)\frac{\eta(R)^{2}}{R^{2}}\right).
\]
Using the explicit solution
\[
\eta(r)=\frac{a}{n}\left(r+(n-1)r\ln\left(\frac{r}{R}\right)\right),
\]
we  then obtain the expression for the energy 
\[
 E[\Bu]=\frac{a^{2}(n-1)}{2n^{2}}|B(0,R)|,
\]
which  naturally agrees with   the   direct 
  computation of  the stored elastic energy   in  (\ref{Epos1}). We
  reiterate  that the classical Clapeyron Theorem does not account for
  such ``work of creation'' of a residually stressed elastic body.
  The reason for the failure of the classical Clapeyron's Theorem in the presence of incompatibility is that it  breaks the homogeneity of degree 2  required for this classical theorem to hold. 
  
  Using Theorem~\ref{th:clapeyron} we can generalize our example to illustrate additional aspects of the configurational loading by  the generalized tractions $\BP^{*}\Bn$.
We now only assume that  there exists a general  ``inelastic strain'' $\BGve_{0}(\Hat{\Bx})$, such that the energy density $W(\Hat{\Bx},\BGve)$ satisfies, for any $\BGve\in\Sym(\bb{R}^{n})$ and any unit vector $\Hat{\Bx}$, the inequality
  \begin{equation}
    \label{Wlb}
    W(\Hat{\Bx},\BGve)\ge c|\BGve-\BGve_{0}(\Hat{\Bx})|^{p},
  \end{equation}
for some $p>0$ and $c>0$. Note that the energy density (\ref{Wdisloc}) with (\ref{tracfree3}) in our previous example is of this type.
Note next that if there exists a Lipschitz stationary equilibrium $\Bu_{0}$ in $\GO$ satisfying both $\BP\Bn=0$ and $\BP^{*}\Bn=0$ on $\dOm$ in the sense of traces, then, according to  (\ref{Clapeyron}) we must have
\[
\int_{\GO}W(\Hat{\Bx},e(\Bu_{0}))d\Bx=0.
\]
But then
\[
\int_{\GO}|e(\Bu_{0})-\BGve_{0}(\Hat{\Bx})|^{p}d\Bx\le\nth{c}\int_{\GO}W(\Hat{\Bx},e(\Bu_{0}))d\Bx=0.
\]
It would then mean that $\BGve_{0}(\Hat{\Bx})=e(\Bu_{0})$ is a compatible strain. 
To make clear in this case the connection between the incompatibility of $\BGve_{0}(\Hat{\Bx})$ and the configurational tractions on the boundary we now  generalize the problem by dropping the assumption that $\BP^{*}\Bn=0$ on $\dOm$. Suppose that $\BGve_{0}(\Hat{\Bx})$ is an incompatible  strain field, and suppose that $\Bu_{0}$ is the minimizer in
\begin{equation}
  \label{m}
 \CM=\min_{\Bu\in H^{1}(\GO;\bb{R}^{n})}\int_{\GO}W(\Hat{\Bx},e(\Bu))d\Bx.
\end{equation}
Observe that due to (\ref{Wlb}), we can regard $ \CM >0$ as a measure of incompatibility of $\BGve_{0}(\Hat{\Bx})$. At the same time, the minimizer $\Bu_{0}$ in (\ref{m}) must be a stationary extremal, satisfying $\BP\Bn=0$. But then, Theorem~\ref{th:clapeyron} gives
\[
 \CM=\nth{n}\int_{\dOm}\BP^{*}\Bn\cdot\Bx\,dS,
\]
which ultimately links the  applied configurational tractions to the measure of  incompatibility of the ``transformation strain'' $\BGve_{0}(\Hat{\Bx})$ expressed by the scalar $\CM$.

\section{Applications of CET in nonlinear elasticity}
\label{sec:app}
In this section we present a few examples of the  
effectiveness of  CET  in various challenging problems of nonlinear
elasticity   rooted in the Calculus of Variations.  We 
remark that while formula (\ref{increm}) holds for any Lipschitz
stationary configuration, the quantities $\BP\Bn$ and $\BP^{*}\Bn$,
understood in the sense of traces cannot be understood pointwise,
i.e., as smooth functions of $\Grad\By(\Bx)$, which is not well-defined on
$\dOm$, in general. In practice, the discontinuities of
$\Grad\By(\Bx)$ are confined to closed nowhere dense sets, like phase
boundaries or shock waves, that have either no intersection with
$\dOm$ or intersect $\dOm$ over a closed, subset of $\dOm$ of zero
surface measure. For these configurations $\BP(\Grad\By(\Bx))$ and
$\BP^{*}(\Grad\By(\Bx))$ can be interpreted pointwise on $\dOm$.
Hence, throughout this section we will assume that $\By(\Bx)$ is a
Lipschitz stationary extremal in $\GO$ that is sufficiently regular on
a \nbh\ of $\dOm$, so that $\Grad\By(\Bx)$ is well-defined on $\dOm$,
at least as an $L^{\infty}(\dOm)$ function.

\subsection{Necessary conditions of metastability}
\label{sub:dbc}
The problem of metastability is central for nonlinear elasticity
theory and, more generally, for vectorial variational problems with
nonconvex energy densities.  In our special setting, it concerns the
possibility of existence of strong local minima  that are not automatically global minima.
To analyze this problem we can take advantage of the fact that CET
gives a formula for the energy in terms of the boundary values of
the solution of the Euler-Lagrange equations and its gradient. The
idea is to  consider the difference of CET representations of energies of two 
competing stationary states and rearrange the terms advantageously. 

As our subsequent analysis shows, a  helpful hint is  to  isolate in the boundary integral an  expression involving a Weierstrass excess function 
\begin{equation}
  \label{Weierstrass}
  \CE(\BF,\BG)=W(\BG)-W(\BF)-\av{W_{\BF}(\BF),\BG-\BF}. 
\end{equation}
This function has proven useful in the Calculus of Variations for
comparing the energies of two different stationary configurations with
the same Dirchlet data \cite{hest48,tahe01,grme08}. The result of the proposed  rearrangement can be formulated in the form of a theorem which relies essentially on the CET. 
\begin{theorem}
  \label{th:aff}
  Assume that the Lipschitz deformations $\By_{1}(\Bx)$ and $\By_{2}(\Bx)$ are
  stationary extremals for the energy functional
  \begin{equation}
    \label{nlelen}
    E[\By]=\int_{\GO}W(\Grad\By(\Bx))d\Bx.
  \end{equation}
Suppose that $\By_{1}(\Bx)$ and $\By_{2}(\Bx)$ are of class $C^{1}$ near $\dOm$ and that
\begin{equation}
  \label{bc}
\By_{1}(\Bx)=\By_{2}(\Bx)=\By_{0}(\Bx),\quad \Bx\in\dOm.
\end{equation}
Then
\begin{equation}
\int_{\GO}\{W(\BF_{2})-W(\BF_{1})\}d\Bx=\nth{n}\int_{\Md\GO}
\{\CE(\BF_{1},\BF_{2})(\Bx,\Bn)+(\BP_{1}-\BP_{2})\Bn\cdot(\BF_{2}\Bx-\By_{0})\}dS(\Bx),
\label{naff1}
\end{equation}
where
\[
\BF_{j}=\Grad\By_{j}(\Bx),\qquad\BP_{j}=W_{\BF}(\Grad\By_{j}(\Bx)),\qquad
j=1,2.
\]
\end{theorem}
\begin{proof}
  Since the energy density is scale-free in the sense of
  (\ref{scalefree}), we can use  the GCT Theorem~\ref{th:clapeyron}, and write 
  \[
\int_{\GO}W(\BF_{i})d\Bx=\nth{n}\int_{\dOm}\{\BP_{i}\Bn\cdot\By_{0}+\BP_{i}^{*}\Bn\cdot\Bx\}dS,
\]
Subtracting these two equations and expanding $\BP^{*}_{j}$ via (\ref{Esh-tensor}) we obtain
\begin{multline*}
  \int_{\GO}\{W(\BF_{2})-W(\BF_{1})\}d\Bx=\nth{n}\int_{\Md\GO}\{
  \CE(\BF_{1},\BF_{2})\Bn\cdot\Bx+
  (\BP_{1}-\BP_{2})\Bn\cdot(\BF_{2}\Bx-\By_{0})+\\
    \av{\BP_{1},\BF_{2}-\BF_{1}}\Bn\cdot\Bx
  +\BP_{1}\Bn\cdot(\BF_{1}-\BF_{2})\Bx\}dS.
\end{multline*}
Since $\By_{1}$ and $\By_{2}$ are of class $C^{1}$ near $\dOm$ and
$\By_{1}-\By_{2}=0$ on $\dOm$, we conclude that there exists a
continuous vector field $\Ba:\dOm\to\bb{R}^{m}$, such that
\begin{equation}
  \label{samebc}
  \BF_{1}-\BF_{2}=\Ba\otimes\Bn.
\end{equation}
Therefore,
\[
\BP_{1}\Bn\cdot(\BF_{1}-\BF_{2})\Bx=(\BP_{1}\Bn\cdot\Ba)\Bn\cdot\Bx=\av{\BP_{1},\BF_{1}-\BF_{2}}\Bn\cdot\Bx,
\]
and  (\ref{naff1})   follows.
\end{proof}
\begin{remark}
Switching $\By_{1}$ and $\By_{2}$ we also have
  \begin{equation}
    \label{naff2}
    \int_{\GO}\{W(\BF_{2})-W(\BF_{1})\}d\Bx=\nth{n}\int_{\Md\GO}
\{-\CE(\BF_{2},\BF_{1})\Bx\cdot\Bn+(\BP_{1}-\BP_{2})\Bn\cdot(\BF_{1}\Bx-\By_{0})\}dS(\Bx).
  \end{equation}
  Moreover, taking the average of (\ref{naff1}) and (\ref{naff2}) we obtain a  symmetric version of the energy increment formula
   \begin{multline*}
      \label{symd}
  \int_{\GO}\{W(\BF_{2})-W(\BF_{1})\}d\Bx= \\
  \nth{n}\int_{\Md\GO}\left\{
\hf(\CE(\BF_{1},\BF_{2})-\CE(\BF_{2},\BF_{1}))\Bx\cdot\Bn+(\BP_{1}-\BP_{2})\Bn\cdot(\lump{\BF}\Bx-\By_{0})
\right\}dS(\Bx),
   \end{multline*}
where $\lump{\BF}=(\BF_{1}+\BF_{2})/2$. Subtracting (\ref{naff1}) and (\ref{naff2}) we also obtain another useful result
    \begin{equation}
      \label{syms}
      \int_{\dOm}(\CE(\BF_{1},\BF_{2})+\CE(\BF_{2},\BF_{1}))\Bx\cdot\Bn\,dS(\Bx)=
\int_{\dOm}(\BP_{2}-\BP_{1})\Bn\cdot(\BF_{2}-\BF_{1})\Bx\,dS(\Bx)
    \end{equation}
  \end{remark}
Observe further that if $W(\BF)$ is rank-one convex, i.e. convex along
any affine straight lines in rank one directions, then we have for all $\Bx\in\dOm$
\[
\CE(\BF_{1},\BF_{2})=W(\BF_{2})-W(\BF_{1})-\av{\BP_{1}, \BF_{2}-\BF_{1}}\ge 0,
\]
due to (\ref{samebc}). We also recall that if $\GO$ is a star-shaped
domain, then we can always choose the origin to be a star point in
$\GO$, due to Remark~\ref{rem:shift}. In that case we also have
then $\Bx\cdot\Bn\ge 0$ for all $\Bx\in\dOm$.

Let us assume that $\By_{2}(\Bx)$ is a global minimizer of the
energy functional (\ref{nlelen}), among all Lipschitz functions
satisfying the \bc s (\ref{bc}) on a star-shaped domain $\GO$, where
the origin is at a star-point in $\GO$. We want to know if there are
any metastable states, i.e., strong local energy minimizers that are not global. Suppose now
that $\By_{1}(\Bx)$ is such a  metastable state. Then $\By_{1}$ and
$\By_{2}$ are stationary Lipschitz extremals. Assume further that $\By_{1}$ and
$\By_{2}$ are of class $C^{1}$ near $\dOm$, so that
Theorem~\ref{th:aff} is applicable. Then, the following inequality holds
\begin{equation}
  \label{R1ineq}
  \int_{\GO}\{W(\BF_{2})-W(\BF_{1})\}d\Bx\ge\nth{n}\int_{\Md\GO}
(\BP_{1}-\BP_{2})\Bn\cdot(\BF_{2}\Bx-\By_{0})\,dS(\Bx).
\end{equation}
Now, if $\By_{1}(\Bx)$ satisfies
\begin{equation}
  \label{nslmcrit}
 \int_{\Md\GO}(\BP_{1}-\BP_{2})\Bn\cdot(\BF_{2}\Bx-\By_{0})dS(\Bx)\ge 0,
\end{equation}
then
\[
0\ge\int_{\GO}\{W(\BF_{2})-W(\BF_{1})\}d\Bx\ge \int_{\Md\GO}(\BP_{1}-\BP_{2})\Bn\cdot(\BF_{2}\Bx-\By_{0})dS(\Bx)\ge 0,
\]
and hence $\By_{1}(\Bx)$ must also be a global minimizer. Since we
don't know a priori what $\By_{1}(\Bx)$ could be, the only way to
guarantee that inequality (\ref{nslmcrit}) holds is to demand that it holds
for any matrix $\BP_{1}$. This, in turn, implies that 
\begin{equation}
  \label{1hom}
  \Grad\By_{2}(\Bx)\Bx=\By_{2}(\Bx),\quad\Bx\in\dOm.
\end{equation}
Under our assumption that $\GO$ is star-shaped and that the origin is
a star-point in $\GO$, a \nbh\ of $\dOm$ in $\GO$ is parametrized by
$\BX=t\Bx$, $t\in[1-\Gd,1]$, $\Bx\in\dOm$. Hence, we can write
$\By_{2}(\BX)=\By_{2}(t\Bx)$. Then, it is easy to see that
(\ref{1hom}) is equivalent to the ``infinitesimal 1-homogeneity'' of
a $C^{1}$ vector field $\By_{2}(\Bx)$ in the vicinity of $\dOm$:
\begin{equation}
  \label{inf1hom}
\By_{2}(t\Bx)=t\By_{2}(\Bx)+o(1-t),\quad\Bx\in\GO,\quad t\to 1^{-}.
\end{equation}
Therefore, the \lhs\ of (\ref{nslmcrit}) can be regarded as a measure
of infinitesimal 1-nonhomogeneity of $\By_{2}(\Bx)$.  The failure of
(\ref{inf1hom}) can be viewed as a practically verifiable necessary
condition of metastability (provided the global minimizer
$\By_{2}(\Bx)$ is known). The failure of inequality (\ref{nslmcrit})
is also a necessary condition that a specific Lipschitz stationary
extremal $\By_{1}(\Bx)$ is a true metastable state. Since affine
functions $\By_{2}(\Bx)=\BF_{0}\Bx$ satisfy (\ref{inf1hom}), we
conclude that in the case of
affine boundary conditions the metastability is absent, at least on
star-shaped domains (see Section~\ref{sub:qcx} below and \cite{tahe03,grtrhard}).

\subsection{Quasiconvex envelope}
\label{sub:qcx} 
We recall that the quasiconvex envelope of the energy density $W(\BF)$ is the largest quasiconvex function that does not exceed $W(\BF)$. It is denoted $QW$ and is often referred to  as quasiconvexification of $W$. There is a formula for $QW$ \cite{Dak08}
\begin{equation}
  \label{QW2}
  QW(\BF)=\inf_{\myatop{\By|_{\Md D}=\BF\Bx}{\By\in W^{1,\infty}(D;\bb{R}^{m})}}\nth{|D|}\int_{D}W(\Grad\By)d\Bx,
\end{equation}
where $D$ can be any open and bounded set, with $|\Md D|=0$. 
The rank-one convex envelope $RW$ is the largest rank-one convex
function\footnote{i.e.,function convex along all straight lines connecting two
points that differ by a rank-one matrix.} that does not exceed $W$. It was shown in \cite{morr52} that every quasiconvex function is rank-one convex and therefore, $RW(\BF)\ge QW(\BF)$. The  question  whether $QW(\BF)=RW(\BF)$  has been answered in the negative in \cite{sver92}, except in the case $m=2$. Nonetheless nontrivial examples where $QW(\BF)\not=RW(\BF)$ are extremely rare, see for instance \cite{grab18}.

While, in many cases the variational problem (\ref{QW2}) has no solutions,
there could be particular domains $D$, where Lipschitz minimizers do
exist \cite{grtrhard}. 
Or next theorem shows how the CET can
furnish somewhat minimal additional conditions that stationary
extremals of (\ref{nlelen}) should satisfy to be minimizers in
(\ref{QW2}). Furthermore, the CET permits us to
transfer the knowledge about quasiconvexity from a single point to the
entire range of $\Grad\By$, provided that rank-one convexity at the values of $\Grad\By$ on the boundary can be verified.
\begin{theorem}
  \label{th:ELA}
Let $\GO\subset{\bb{R}^{d}}$ be a star-shaped domain with Lipschitz boundary. Assume that
\begin{itemize}
\item[(BVP)] For a given $\BF_{0}\in\bb{R}^{m\times n}$  the function  $\bra{\By}\in W^{1,\infty}(\GO;\bb{R}^{m})$ solves 
\begin{equation}
  \label{ELN1}
  \Div\BP(\Grad\bra{\By})=\Bzr,\quad\Div\BP^{*}(\Grad\bra{\By})=\Bzr,
\end{equation}
in the sense of distributions, and satisfies $\bra{\By}(\Bx)=\BF_{0}\Bx$ for all $\Bx\in\dOm$;
\item[(REG)] $\bra{\By}$ is of class $C^{1}$ near $\dOm$;
\item[(RCX)] $W(\Grad\bra{\By}(\Bx))=RW(\Grad\bra{\By}(\Bx))$ for all $\Bx\in\dOm$;
\item[(Q=R)] $RW(\BF_{0})=QW(\BF_{0})$.
\end{itemize}
Then
\begin{enumerate}
\item[(i)] $\displaystyle RW(\BF_{0})=QW(\BF_{0})=\nth{|\GO|}\int_{\GO}W(\Grad\bra{\By})d\Bx$, i.e. $\bra{\By}(\Bx)$ is the global minimizer in (\ref{QW2}).
\item[(ii)] $QW(\Grad\bra{\By}(\Bx))=W(\Grad\bra{\By}(\Bx))$ for a.e. $\Bx\in\GO$.
\end{enumerate}
\end{theorem}
\begin{proof}
  Assumption (\textit{BVP}) implies that the GCT is applicable to the
  stationary equilibrium $\bra{\By}(\Bx)$ of (\ref{QW2}), and we obtain, by
  virtue of (\textit{REG}), that
  \[
\Grad\By(\Bx)=\BF_{0}+\Ba(\Bx)\otimes\Bn,\quad\Bx\in\dOm,
\]
for some $\Ba\in C(\dOm;\bb{R}^{m})$, and therefore,
\begin{equation}
  \label{ClapRW}
  \int_{\GO}W(\Grad\bra{\By})d\Bx=\nth{n}\int_{\dOm}(W(\BF_{0}+\Ba\otimes\Bn)-
  \BP(\BF_{0}+\Ba\otimes\Bn)\Bn\cdot\Ba)(\Bn\cdot\Bx)dS. 
\end{equation}

Next, we appeal to \cite[Lemma~4.2]{grtrnc}.
We quote the part of the lemma needed for the proof.
\begin{lemma}
  \label{lem:r1cx}
Let $V(\BF_{0})$ be a rank-one convex function such that $V(\BF_{0})\le W(\BF_{0})$. Let
\[
\CA_{V}=\{\BF_{0}\in\bb{R}^{m\times n}:W(\BF_{0})=V(\BF_{0})\}.
\]
Then for every $\BF_{*}\in\CA_{V}$, $\Bb\in\bb{R}^{m}$, and $\Bm\in\bb{R}^{n}$
\begin{equation}
  \label{coolineq}
  V(\BF_{*}+\Bb\otimes\Bm)\ge W(\BF_{*})+\BP(\BF_{*})\Bm\cdot\Bb.
\end{equation}
\end{lemma}
We now apply Lemma~\ref{lem:r1cx} by choosing $V(\BF)=RW(\BF)$,
$\BF_{*}=\Grad\bra{\By}(\Bx)=\BF_{0}+\Ba\otimes\Bn$, and $\Bb\otimes\Bm=-\Ba(\Bx)\otimes\Bn$,
for each $\Bx\in\dOm$. The
assumption (\textit{RCX}) implies that Lemma~\ref{lem:r1cx} is applicable, and therefore,
\begin{equation}
  \label{QWineq}
RW(\BF_{0})\ge W(\BF_{0}+\Ba\otimes\Bn)-\BP(\BF_{0}+\Ba\otimes\Bn)\Bn\cdot\Ba.
\end{equation}
Finally, we use the assumption that $\GO$ is star-shaped. If we choose the
origin at the star point, then the function $\Bn(\Bx)\cdot\Bx$ is always
non-negative at all points on $\dOm$. Therefore, inequality (\ref{QWineq})
together with the identity (\ref{ClapRW}) delivered by the GCT,
implies
\[
\int_{\GO}W(\Grad\bra{\By})d\Bx\le \frac{RW(\BF_{0})}{n}\int_{\dOm}(\Bn\cdot\Bx)dS=|\GO|RW(\BF_{0}).
\]
Since $|\GO|QW(\BF_{0})$ is the minimal value of the functional in (\ref{QW2}) we
obtain
\begin{equation}
  \label{RWQW}
  |\GO|RW(\BF_{0})\ge\int_{\GO}W(\Grad\bra{\By})d\Bx\ge|\GO|QW(\BF_{0}). 
\end{equation}
By the assumption (\textit{Q=R}) we must have equality in both inequalities in (\ref{RWQW}) and conclude that $\bra{\By}(\Bx)$ is the global minimizer in (\ref{QW2}), proving (i). Property (ii) is a necessary condition for any strong local minimizer \cite{tahe02}, and hence must hold as well.
\end{proof}
Theorem~\ref{th:ELA} provides a new tool for establishing quasiconvexity,
or, at least, estimating the quasiconvex envelope. We remark that
Theorem~\ref{th:ELA} yields particularly powerful results when the
extremal $\bra{\By}$ has radial symmetry and $\GO=\bb{R}^{n}$, as we
will see in the next section.

\subsection{Probing the binodal}
It is well-known \cite{ball7677} that if $\By(\Bx)$ is a strong local
minimizer of (\ref{non-param}), then the function $\BF\mapsto
W(\Bx,\By(\Bx),\BF)$ must be quasiconvex at $\BF=\Grad\By(\Bx)$, for
a.e. $\Bx\in\GO$, i.e., satisfy the inequality
\begin{equation}
  \label{qcxF}
 \int_{D}\{W(\BF+\Grad\BGf(\Bx))-W(\BF)\}d\Bx\ge 0\quad\forall\BGf\in W_{0}^{1,\infty}(D;\bb{R}^{m}), 
\end{equation}
where $D$ is any domain in $\bb{R}^{n}$ with $|\Md D|=0$. Thus, the gradient of any metastable configuration must avoid the binodal region
\[
\CB=\left\{\BF\in\bb{R}^{m\times n}:\exists\BGf\in W_{0}^{1,\infty}(D;\bb{R}^{m}): \int_{D}\{W(\BF+\Grad\BGf(\Bx))-W(\BF)\}d\Bx<0\right\}
\]
of the phase space. Its boundary, $\Md\CB$, is called the \emph{binodal} \cite{grtrmms,grtrnc,grtrlhqcx,grtrpcx,grtrsolid,grtrhard,grtrliqnuc}.

 Since the function is quasiconvex, \IFF its binodal region is empty,  
the binodal region can then be  characterized by the inequality \cite{huss95}
\[
\CB=\{\BF\in\bb{R}^{m\times n}:QW(\BF)<W(\BF)\}.
\]
 Note also that if $W(\BF_{0})$ is strictly quasiconvex at $\BF_{0}\in\bb{R}^{m\times n}$, then $\bra{\By}(\Bx)=\BF_{0}\Bx$ is the unique solution of (\ref{QW2}) \cite{tahe03}. If $\BF_{0}$ lies on the binodal, we expect the emergence of nontrivial equilibria in unbounded domains, where the volume fraction of the region where $\Grad\bra{\By}$
is a finite distance from $\BF_{0}$ to be zero, even though in most finite domains we still expect that $\bra{\By}(\Bx)=\BF_{0}\Bx$ be the unique solution.
 To follow this path  one can, for instance,  try to  find radial solutions of (\ref{ELN1}) in the entire space. As we show below, in case of the success,  CET can   deliver a formula for $QW(\BF)$ for a range of $\BF$. 

Suppose $m=n$, and $W(\BF)$ is objective and isotropic \[W(\BF)=w(v_{1},\ldots,v_{n}),\] where $v_{j}$ are
the singular values of $\BF$ and the function $w$ does not change after a permutation of its arguments.
We look for a radial solution $\By(\Bx)=\eta(r)\Hat{\Bx}$ of (\ref{ELN1}).
We first observe that
\[
\Grad\By=\eta'(r)\tns{\Hat{\Bx}}+\frac{\eta(r)}{r}(\BI_{n}-\tns{\Hat{\Bx}}).
\]
Hence, $\BF=\Grad\By$ has a singular value $|\eta'(r)|$ and $n-1$ singular values $|\eta(r)/r|$.
We will assume that $\eta(r)\ge 0$ and $\eta'(r)\ge 0$. In that case we define two functions
\[
  w_{1}(r)=\dif{w}{v_{1}}\left(\eta'(r),\frac{\eta(r)}{r},\ldots, \frac{\eta(r)}{r}\right),\quad
  w_{2}(r)=\dif{w}{v_{2}}\left(\eta'(r),\frac{\eta(r)}{r},\ldots, \frac{\eta(r)}{r}\right).
\]
The Piola stress tensor is \cite[p.~564]{ball82}
\[
\BP=w_{1}(r)\tns{\Hat{\Bx}}+w_{2}(r)(\BI_{n}-\tns{\Hat{\Bx}}).
\]
Therefore, $\eta(r)$ must solve the nonlinear second order ODE
\begin{equation}
  \label{rODE}
w_{11}(r)\eta''(r)+(n-1)w_{12}(r)\left(\frac{\eta(r)}{r}\right)'+\frac{n-1}{r}(w_{1}(r)-w_{2}(r))=0,
\end{equation}
where
\[
  w_{11}(r)=\hess{w}{v_{1}}\left(\eta'(r),\frac{\eta(r)}{r},\ldots, \frac{\eta(r)}{r}\right),\quad
  w_{12}(r)=\mix{w}{v_{1}}{v_{2}}\left(\eta'(r),\frac{\eta(r)}{r},\ldots, \frac{\eta(r)}{r}\right).
\]
If $\eta'(r)$ suffers a jump discontinuity at $r=r_{0}$, then for (\ref{ELN1}) to hold we must have
\[\jump{w_{1}}(r_{0})=0,\quad\jump{w}(r_{0})=w_{1}(r_{0})\jump{\eta'}.\]
Typically, we would look for a solution of the form 
\begin{equation}
  \label{radsol}
  \By(\Bx)=
  \begin{cases}
    f_{0}r\Hat{\Bx}, &|\Bx|<1,\\
    \eta(r)\Hat{\Bx},&|\Bx|>1
  \end{cases}
\end{equation}
In this case we must have
\[
  w_{1}(\eta'(1^{+}),f_{0},\ldots,f_{0})=w_{1}(f_{0},f_{0},\ldots,f_{0})
\]
\[
  w(\eta'(1^{+}),f_{0},\ldots,f_{0})=w(f_{0},f_{0},\ldots,f_{0})+w_{1}(f_{0},f_{0},\ldots,f_{0})
  (\eta'(1^{+})-f_{0}).
\]
These can be regarded as equations determining the values of $f_{0}$ and $\eta'(1^{+})$.
Ultimately, the radial solution $\eta(r)$ will be fixed by
\[
f_{\infty}=\lim_{r\to\infty}\frac{\eta(r)}{r}=\lim_{r\to\infty}\eta'(r).
\]
Moreover, the indicial equation at $r=\infty$ yields the asymptotics
\begin{equation}
  \label{etainf}
  \eta(r)=f_{\infty}r+\frac{A}{r^{n-1}}+O(r^{-n}),\text{ as }r\to\infty.
\end{equation}
for  the solutions of \eqref{rODE}. To derive it, we need to substitute the ansatz
\begin{equation}
  \label{etasym}
  \eta(r)=f_{\infty}r+Ar^{-\Ga}+O(r^{-\Ga-1})
\end{equation}
into \eqref{rODE}. It is then clear that the asymptotics at $r=\infty$ of various terms in \eqref{rODE} are as follows:
\[
  w_{11}\eta''=w_{11}^{\infty}\frac{\Ga(\Ga+1)A}{r^{\Ga+2}}+o(r^{-(\Ga+2)}),\quad
 w_{12}\left(\frac{\eta}{r}\right)'=-w_{12}^{\infty}\frac{(\Ga+1)A}{r^{\Ga+2}}+o(r^{-(\Ga+2)}),
\]
\[
  w_{1}=w_{1}^{\infty}-w_{11}^{\infty}\frac{\Ga A}{r^{\Ga+1}}
  +(n-1)w_{12}^{\infty}\frac{A}{r^{\Ga+1}}+o(r^{-(\Ga+1)}),
\]
\[
  w_{2}=w_{2}^{\infty}-w_{21}^{\infty}\frac{\Ga A}{r^{\Ga+1}}+w_{22}^{\infty}\frac{A}{r^{\Ga+1}}+
  (n-2)w_{23}^{\infty}\frac{A}{r^{\Ga+1}}+o(r^{-(\Ga+1)}),
\]
where
\[
w_{i}^{\infty}=\dif{w}{v_{i}}(f_{\infty},\ldots,f_{\infty}),\quad w_{ij}^{\infty}=\mix{w}{v_{i}}{v_{j}}(f_{\infty},\ldots,f_{\infty}).
\]
Thus, the indicial equation for $\Ga$ is, taking into account that
\[
  w_{1}^{\infty}=w_{2}^{\infty},\quad w_{23}^{\infty}=w_{21}^{\infty}=w_{12}^{\infty},\quad w_{22}^{\infty}=w_{11}^{\infty},
\]
due to the permutation symmetry of $w(v_{1},\ldots,v_{n})$, 
\[
  w_{11}^{\infty}\Ga(\Ga+1)-(n-1)(\Ga+1)w_{12}^{\infty}+(n-1)(\Ga+1)(w_{12}^{\infty}-w_{11}^{\infty})=0.
\]
Dividing both sides by $\Ga+1$ (since $\Ga=-1$ corresponds to the first term $f_{\infty}r$ in the asymptotics (\ref{etasym})) we obtain
\[
w_{11}^{\infty}\Ga-(n-1)w_{12}^{\infty}+(n-1)(w_{12}^{\infty}-w_{11}^{\infty})=0.
\]
This becomes $w_{11}^{\infty}(\Ga-n+1)=0$, and therefore, $\Ga=n-1$, as claimed.

The Theorem below  demonstrates   how in this case  the GCT     transforms the knowledge of quasiconvexity at infinity (the boundary of all-space)  into an explicit formula for $QW(f\BI_{n})$ for all $f\in[f_{0},f_{\infty}]$.

\begin{theorem}
  \label{th:QWrad}
Suppose that $\eta\in C^{2}([1,+\infty))$ is such that
\begin{itemize}
\item[(i)] $\displaystyle\bra{\By}(\Bx)=
  \begin{cases}
    f_{0}\Bx, &|\Bx|<1,\\
    \eta(|\Bx|)\Hat{\Bx},&|\Bx|>1
  \end{cases}
$ solves (\ref{ELN1}) in $\bb{R}^{n}$ in the sense of distributions.
\item[(ii)] $\displaystyle f_{\infty}=\lim_{r\to\infty}\frac{\eta(r)}{r}=\lim_{r\to\infty}\eta'(r)$ is such that
$W(f_{\infty}\BI_{n})=QW(f_{\infty}\BI_{n})$.
\end{itemize}
Then 
\begin{itemize}
\item[(a)] $W(f_{0}\BI_{n})=QW(f_{0}\BI_{n})$;
\item[(b)] for all $r\ge 1$
\begin{equation}
  \label{QWR}
    W\left(\eta'(r)\tns{\Hat{\Bx}}+\frac{\eta(r)}{r}(\BI_{n}-\tns{\Hat{\Bx}})\right)=
  QW\left(\eta'(r)\tns{\Hat{\Bx}}+\frac{\eta(r)}{r}(\BI_{n}-\tns{\Hat{\Bx}})\right). 
\end{equation}
\item[(c)] for every $r\ge 1$
  \begin{equation}
    \label{QWfI}
    QW\left(\frac{\eta(r)}{r}\BI_{n}\right)=w_{1}(r)\frac{\eta(r)}{r}+w(r)-\eta'(r)w_{1}(r).
  \end{equation}
\end{itemize}
\end{theorem}
\begin{proof}
According to (\ref{qcxF}), quasiconvexity is defined in terms of a bounded domain. In \cite{grtrmms} we have linked quasiconvexity to test functions in all-space. 
According to \cite[Theorem~4.6]{grtrmms} the existence of the radial solution $\bra{\By}(\Bx)$ of (\ref{ELN1}) as in assumption (i) implies that
\[
\int_{\bb{R}^{n}}\CE(f_{\infty}\BI_{n},\Grad\bra{\By})d\Bx=0,
\]
where $\CE(\BF,\BG)$ is the Weierstrass excess function
(\ref{Weierstrass}), since
\[
  \Grad\bra{\By}-f_{\infty}\BI_{n}\in\CS=\{\BGf\in W^{1,\infty}(\bb{R}^{n};\bb{R}^{n}):
  \Grad\BGf\in L^{2}(\bb{R}^{n}; \bb{R}^{m\times n})\},
\]
according to (\ref{etainf}). Then,
by \cite[Theorem~3.10]{grtrmms}, the assumption (ii) implies
\[
\int_{\bb{R}^{n}}\CE(f_{\infty}\BI_{n},\Grad\By)d\Bx\ge 0
\]
for all $\By(\Bx)$, such that $\Grad\By(\Bx)-f_{\infty}\BI_{n}\in\CS$.
In that case conclusions (a) and (b) follow from  \cite[Theorem~4.3]{grtrmms}.

To prove part (c) we observe that radial solutions have the property that
$\bra{\By}(\Bx)=f_{R}\Bx$ for all $\Bx\in\Md B(0,R)$, where $f_{R}=\eta(R)/R$. Hence, by the already established conclusion (b), Theorem~\ref{th:ELA}  applies and says that
\begin{equation}
  \label{QWrad}
  |B(0,R)|QW\left(\frac{\eta(R)}{R}\BI_{n}\right)=\int_{B(0,R)}W(\Grad\bra{\By})d\Bx.
\end{equation}
We complete the proof by the application of the Clapeyron theorem:
\[
  \BP\Bn=w_{1}\Hat{\Bx},\quad
  \BP^{*}=w\BI_{n}-\eta'w_{1}\tns{\Hat{\Bx}}+w_{2}\frac{\eta}{r}(\BI_{n}-\tns{\Hat{\Bx}})
\]
Thus, according to (\ref{Clapeyron}),
\[
  \int_{B(0,R)}W(\Grad\bra{\By})d\Bx=\nth{n}\int_{\Md B(0,R)}\left\{w_{1}\Hat{\Bx}\cdot\eta\Hat{\Bx}+
    (w-\eta'w_{1})\Hat{\Bx}\cdot R\Hat{\Bx}\right\}dS.
\]
We conclude that
\[
QW\left(\frac{\eta(R)}{R}\BI_{n}\right)=w_{1}(R)\frac{\eta(R)}{R}+w(R)-\eta'(R)w_{1}(R).
\]
The theorem is now proved.
\end{proof}

\section{Conclusions}
\label{sec:conc}
In this paper we revisited the classical Clapeyron's Theorem (CT)
(\ref{Clap}) of  the  linear elasticity theory
which expresses the energy of an equilibrium configuration in terms of
the work of physical forces on the boundary.  We interpreted this
engineering result in the framework of Noether formula (\ref{increm}) for general
variation of the graph of a stationary extremal.
This new perspective allowed us to derive  various  nonlinear analogs
of  the CT  by revealing their  links  with partial
symmetries of the corresponding Lagrangians, such a $p$-homogeneity
and scale invariance. The resulting formulas  express bulk energy in terms of
the generalized surface work that includes not only physical but also
configurational forces.

An important   mechanical aspect   of our approach  is the possibility
to  differentiate between the part of the elastic energy stored due to
the action of applied physical loads  and the ``cold work type'' part
of the energy emerging due to accumulation of elastic
incompatibility. We showed further that when scale invariance is
combined with $p$-homogeneity,  the equilibrium energy can be expressed
through the work of either only physical or only configurational
forces, except when $p=n$, in which case the work of configurational
forces is always zero.

Our paper reveals  that while the  classical CT 
and the CET (Theorem~\ref{th:clapeyron}) are
independent relations, coming from different variational
symmetries,  (\ref{scalingxy}) and (\ref{uscale}), both can be derived
from Theorem~\ref{th:pGCT}. One of them directly, as a quadratic
function of the strain, the other, via
reformulating the nonparametric variational functional
(\ref{non-param}) as a parametric one, which is always $n$-homogeneous.

More generally,  the proposed  broader reading of the CT  allowed us to link it to various seemingly unrelated  
results in nonlinear elasticity   and the possibility to express  the
minimal value of the equilibrium energy through the Weierstrass excess
function.  Moreover, we provided compelling evidence    that    the
CET can open new    ways  of  addressing   important problems of
material stability, in particular,   regarding the possibility to
rule out the presence of strong local minimizers that are not global
in  in the case of hard device loading.  Of significant  interest are
the potential applications of  the CET to elastodynamics where it
naturally accounts for the inevitable  formation of  shock waves. We
deal with some aspects of these particular applications  in
\cite{grtreldyn},  accounting for the fact that extremals are not
necessarily stationary extremals of the action functional, in which
case CET become  a global inequality.

It is also important to mention that   since the concept of a partial
variational symmetry is rather general, Noether's formula used to
obtain  various nonlinear versions of the  classical CT
applies to  a broad class of  variational problems not at all related
to nonlinear elasticity theory. For instance,  the proposed approach
can offer  a unified interpretation to various   Rellich and Pohozaev
type identities  which have proved   useful in several
reaction-diffusion type theories.  These issues will be addressed in a
companion  paper \cite{grtrpoho}, where the CET is applied to problems
involving semi-linear equations, including p-Laplacian and its various generalizations.  Finally, we show in \cite{grtrvoid} how the CET can be used in important engineering problems related to fracture and cavitation.

\section{Acknowledgments} 
YG was supported by the National Science Foundation under Grant
No. DMS-2305832. The work of LT was supported by the French grant
ANR-10-IDEX-0001-02 PSL. The authors appreciate support from the
Oberwolfach Fellows program.

\def\cprime{$'$} \ifx \cedla \undefined \let \cedla = \c \fi\ifx \cyr
  \undefined \let \cyr = \relax \fi\ifx \cprime \undefined \def \cprime
  {$\mathsurround=0pt '$}\fi\ifx \prime \undefined \def \prime {'}
  \fi\def\Ya{Ya}

\end{document}